\documentclass[a4paper,11pt]{article}
\pdfoutput=1
\usepackage{jheppub}
\renewcommand\acknowledgments{\section*{Acknowledgements}}
\usepackage[T1]{fontenc}
\usepackage{amsmath,amssymb,amsthm,enumerate, xcolor,mathtools,cleveref,tikz-cd,mleftright}
\usetikzlibrary{decorations.markings}
\tikzset{middlearrow/.style={
        decoration={markings,
            mark= at position 0.5 with {\arrow{#1}} ,
        },
        postaction={decorate}
    }
}
\usepackage[british]{babel}
\usepackage{hyperref}
\usepackage{lmodern}
\usepackage[final,spacing]{microtype}

\newcommand{\fL}{\ensuremath{\mathfrak{L}}}
\newcommand{\cm}{\circ}

\DeclareMathOperator{\id}{id}

\newcommand{\ket}[1]{|#1\rangle}
\newcommand{\bra}[1]{\langle#1|}

\newtheorem{theorem}{Theorem}
\theoremstyle{definition}
\newtheorem{definition}{Definition}
\hypersetup {
    pdftitle = {Out-of‐time‐order asymptotic observables are quasi‐isomorphic to time-ordered amplitudes},
    pdfauthor = {Leron Borsten, David Simon Henrik Jonsson, and Hyungrok Kim},
    pdfsubject = {hep-th, math-ph},
    pdflang = {en-GB},
    pdfkeywords = {Schwinger–Keldysh formalism, L∞‐algebra, non‐time‐ordered correlators, asymptotic observables, Batalin–Vilkovisky formalism}
}

\title{Out-of‐time‐order asymptotic observables are quasi‐isomorphic to time-ordered amplitudes}

\author[a]{Leron Borsten,}
\author[a]{D. Simon H. Jonsson,}
\author[a,1]{Hyungrok Kim\note{Corresponding author.}}

\affiliation[a]{Department of Physics, Astronomy and Mathematics, University of Hertfordshire, Hatfield, Hertfordshire, AL10 9AB, United Kingdom}
\emailAdd{l.borsten@herts.ac.uk}
\emailAdd{d.jonsson@herts.ac.uk}
\emailAdd{h.kim2@herts.ac.uk}

\abstract{
  Asymptotic observables in quantum field theory beyond the familiar $S$-matrix have recently attracted much interest, for instance in the context of gravity waveforms. Such observables can be understood in terms of Schwinger--Keldysh-type `amplitudes' computed by a set of modified Feynman rules involving cut internal legs and external legs labelled by time-folds. 

  In parallel, a \emph{homotopy-algebraic} understanding of perturbative quantum field theory has emerged in recent years. In particular, passing through homotopy transfer, the $S$-matrix of a perturbative quantum field theory can be understood as the minimal model of an associated (quantum) $L_\infty$-algebra. 

  Here we bring these two developments together. In particular, we show that Schwinger--Keldysh amplitudes are naturally encoded in an \(L_\infty\)-algebra, similar to ordinary scattering amplitudes. As before, they are computed via homotopy transfer, but using deformation-retract data that are \emph{not} canonical (in contrast to the conventional $S$-matrix). We further show that the \(L_\infty\)-algebras encoding Schwinger--Keldysh amplitudes and ordinary amplitudes are quasi-isomorphic (meaning, in a suitable sense, equivalent). This entails a set of recursion relations that enable one to compute Schwinger--Keldysh amplitudes in terms of ordinary amplitudes or vice versa.
  }

\begin{document}
\maketitle
\flushbottom

\section{Introduction and summary}
We cast the panoply of generalised asymptotic observables arising in quantum field theory (whenever the $S$-matrix is well-defined) in terms of homotopy algebras and homotopy transfer. The former is of much current interest in the scattering amplitudes (meets gravitational waves) community, while the latter is a natural mathematical language for perturbative quantum field theory. By bringing them together, we hope to shed new light on both facets. To orientate our diverse intended audiences we begin with an expanded introduction, briefly reviewing generalised asymptotic observables in \cref{Generalised_asymptotic_observables}, $L_\infty$-algebras and quantum field theory in \cref{algebras_qft_scattering}, and $L_\infty$-algebras and scattering amplitudes in \cref{algebras_and_scattering_amplitudes}, before summarising the results of the present contribution in \cref{algebras_generalised_observables}.

\subsection{Generalised asymptotic observables}\label{Generalised_asymptotic_observables}
Generic quantum field theories admit a rich set of asymptotic observables that go beyond the familiar \(S\)-matrix as recently elucidated in \cite{Caron-Huot:2023vxl}. A motivating and well-known class of such asymptotic observables are waveforms, which typically involve the expectation value of some operator $\mathcal{O}$ (e.g.~the gravitational field) in the asymptotic future due to a scattering of some states $\ket{\psi}_{\text{in}}$ (e.g.~two black holes) propagating from the far past. Using the Kosower--Maybee--O'Connell formalism \cite{Kosower:2018adc}, classical waveforms of observational relevance can be extracted from scattering-amplitude-like forms using the \emph{in-in} prescription,
\begin{equation}\label{waveform}
  \langle \mathcal{O} \rangle = {}_{\text{in}}\bra{\psi}\mathcal{O} \ket{\psi}_{\text{in}}.
\end{equation}
The appearance of the in-in formalism is eminently reasonable since the expectation value of the observable cares not of the ultimate fate of the in-going states.

For suitable in-states (e.g.~wave packets over free plane-wave states) and observables, \eqref{waveform} can be computed in terms of   free-vacuum $\ket{0}$ matrix elements,
\begin{equation}
    {}_{\text{in}}\bra{\psi}\mathcal{O} \ket{\psi}_{\text{in}} \sim \int f(q_1,q_2, p_1,p_2)\langle 0|a_{q_1}^{\text{in}} a_{q_2}^{\text{in}} b^{\text{out}} a_{p_1}^{ \text{in}\dagger} a_{p_2}^{\text{in}\dagger}| 0\rangle+ \text{c.c.},
\end{equation}
where $a, a^\dagger$ and $b, b^\dagger$ are the creation--annihilation operators associated to the in-state $\ket{\psi}_{\text{in}}$ and out-operator $\mathcal{O}$, respectively\footnote{We have thrown all other data into $f$, which includes the wave packets, polarisation tensors and so forth.} \cite{Cristofoli:2021vyo}.

The key object
\begin{equation}\label{5point}
    \langle 0|a_{q_1}^{\text{in}} a_{q_2}^{\text{in}} b^{\text{out}} a_{p_1}^{ \text{in}\dagger} a_{p_2}^{\text{in}\dagger}| 0\rangle
\end{equation}
is close in form to a conventional 5-point $S$-matrix element, except that the in-state annihilation operators, $a^{\text{in}}$, appear in place of the expected out-state annihilation operators, $a^{\text {out}}$. In fact, \eqref{5point} can be expressed in terms of the familiar 5-point amplitude together with additional \emph{cut amplitudes} \cite{Caron-Huot:2023vxl}. It is a generalised, or out-of-time-order, scattering amplitude in the sense that it exclusively involves asymptotic creation--annihilation operators and the free vacuum, but with the usual stipulation that in-operators appear exclusively to the right of the out-operators relaxed. The central contention of \cite{Caron-Huot:2023vxl} is that \eqref{5point} is but one example of a wealth of generalised/out-of-time-order scattering amplitudes that encode  measurable observables of interest. Classical gravitational waveforms relevant to detectors such as LIGO (\cite{LIGOScientific:2016aoc}) provide highly topical  examples.

A general prescription for the computation of such generalised scattering amplitudes was also given in \cite{Caron-Huot:2023vxl}. Essentially, they arise from a generalised Lehmann--Symanzik--Zimmermann (LSZ) reduction of out-of-time-order correlators in the Schwinger--Keldysh formalism\footnote{See \cite{Chou:1984es, Haehl:2016pec,BenTov:2021jsf} for reviews of  the Schwinger--Keldysh formalism.}. This perspective yields a Feynman-diagrammatic approach in which the expected rules are supplemented by \emph{cut propagators} of the form \(2\pi\mathrm i\delta(p^2-m^2)\Theta(p^0)\) that connect the time-folds of the Schwinger--Keldysh formalism \cite{Caron-Huot:2023vxl}. Of course, out-of-time-order correlators themselves have a long and important history, with applications ranging over condensed matter \cite{larkin1969quasiclassical, Sachdev_1993, cond-mat/0412296,cond-mat/0506130,Kita_2010}, quantum computation, black holes and quantum chaos \cite{larkin1969quasiclassical, Hayden:2007cs, Hosur:2015ylk, Shenker:2013pqa, Maldacena_2016,Heyl_2018}, cosmology \cite{Haque:2020pmp}, and the AdS/CFT correspondence \cite{Kitaev:2017awl, Maldacena:2016hyu, Caron-Huot:2022lff}. It is natural, therefore, to expect their LSZ reduction to entail nontrivial insights.

This paper shows that the generalised scattering amplitudes naturally package themselves into an \(L_\infty\)-algebra with a deformed homotopy retract, generalising the relationship between conventional scattering amplitudes and \(L_\infty\)-algebras, the subject we turn to now.

\subsection{Homotopy algebras and quantum field theory}\label{algebras_qft_scattering}
Ordinary scattering amplitudes admit a natural presentation \cite{Jurco:2018sby,Jurco:2019bvp,Arvanitakis:2019ald,Macrelli:2019afx,Jurco:2019yfd,Jurco:2020yyu} in terms of an $L_\infty$-algebra and homotopy transfer (see e.g.~\cite{Loday:2012aa} for a review of the latter). The starting point of this observation is the connection between the Batalin--Vilkovisky (BV) approach to quantising gauge theories and \emph{homotopy} Lie algebras or $L_\infty$-algebras. Here we give a lightning review of this correspondence; see e.g.~\cite{Jurco:2018sby} for a detailed treatment. This discussion will be rather heuristic; our precise $L_\infty$-algebra conventions are given in \cref{conventions}.

Let us begin by briefly recalling the (classical) BV formalism \cite{Batalin:1977pb,Batalin:1981jr,Batalin:1984jr,Batalin:1984ss,Batalin:1985qj,Schwarz:1992nx}. See \cite{Henneaux:1992, Gomis:1994he, Jurco:2018sby} for pedagogical reviews.
For a gauge theory with physical (and potentially also auxiliary) fields $\varphi$ and ghost fields $c$ (corresponding to the gauge symmetries) one adds a set of antifields $\varphi^+, c^+$ and forms the classical BV action
\begin{equation}
    S_{\text{BV}}[\varphi, c, \varphi^+, c^+]= S_{\text{classical}}[\varphi] + S[\varphi, c, \varphi^+, c^+],
\end{equation}
where $S_{\text{classical}}[\varphi]$ is the usual classical action. Yang--Mills theory on a $d$-dimensional spacetime manifold $M$ with colour Lie algebra $\mathfrak{g}$ provides the paradigmatic example. It consists of a $\mathfrak{g}$-valued 0-form ghost field $c$, gauge potential 1-form $A$, $(d-1)$-form antifield $A^+$ and $d$-form antifield ghost $c^+$, with BV action given by
\begin{equation}
  \begin{split}\label{YMBV}
    S_{\text{BV}}[A, c, A^+, c^+]&=\mathrm{tr}\int_M\frac{1}{2} F \wedge \star F - A^{+}\wedge D c+\frac{1}{2} c^{+}\wedge [c, c]\\
    &=\mathrm{tr}\int_M\frac{1}{2} \mathrm{d}A \wedge \star \mathrm{d}A+\mathrm{d}A \wedge \star [A,A] +\frac{1}{2} [A,A]\wedge \star [A,A] \\ &\quad\quad\quad\quad- A^{+}\wedge D c+\frac{1}{2} c^{+}\wedge [c, c].
  \end{split}
\end{equation}

Further introducing antighost (anti)fields $\bar c^{(+)}$ and Nakanishi--Lautrup (anti)fields $b^{(+)}$, the antifields may be eliminated through the familiar gauge-fixing fermion to yield a Becchi--Rouet--Stora--Tyutin (BRST) action from which the complete set of Feynman rules may be derived, ready to be used in the computation of scattering amplitudes.

What does this have to do with $L_\infty$-algebras? In a nutshell, the set of (anti)fields\footnote{After a degree shift.} forms a graded vector space $V=\bigoplus_n V^n$ that carries a cyclic $L_\infty$-algebra structure. In the current example, the BV (anti)fields of Yang--Mills theory in \eqref{YMBV} belong to
\begin{equation}
  \begin{split}
    V_{\text{YM}} &= V_{\text{YM}}^0\oplus V_{\text{YM}}^1\oplus V_{\text{YM}}^2\oplus V_{\text{YM}}^3\\
    &=\Omega^0(M, \mathfrak{g})\oplus\Omega^1(M, \mathfrak{g})\oplus\Omega^{d-1}(M, \mathfrak{g})\oplus\Omega^d(M, \mathfrak{g}),
  \end{split}
\end{equation}where $\Omega^p(M, \mathfrak{g})$ is the space of $\mathfrak{g}$-valued $p$-forms on $M$. The $L_\infty$-algebra and cyclic  structures on $V_{\text{YM}}$ are determined by the corresponding equations of motion and BV action, respectively. In fact, every perturbative Lagrangian quantum field theory corresponds precisely to a cyclic $L_\infty$-algebra in this manner. Thus, let us briefly summarise the relevant basics of $L_\infty$-algebras so that we may draw out this correspondence. An expository account is given in \cite{Jurco:2018sby}.

Recall that an ordinary Lie algebra consists of a vector space $V$ equipped with an antisymmetric bilinear product, the Lie bracket $[x_1, x_2]\eqqcolon \mu_2(x_1, x_2)$, where the latter notation will become clear momentarily. An $L_\infty$-algebra generalises these data to a $\mathbb{Z}$-graded vector space $V=\bigoplus_n V^n$, where $n$ is the degree of $V^n$, equipped with graded-antisymmetric $k$-linear products $\mu_k(x_1, x_2,\dotsc,x_k)$, which one ought to think of as \emph{higher} Lie brackets\footnote{This generalisation to homotopy Lie algebras is a cornerstone of higher gauge theory and generalised symmetries in quantum field theory. For reviews, see \cite{Borsten:2024gox} and references therein.}. While $[-,-]$ obeys the Jacobi identity on the nose, the $\mu_k(-, -,\dotsc, -)$ obey \emph{higher} Jacobi identities merely up to \emph{homotopies} given by the higher brackets themselves. For example, assuming for simplicity $x_i$ all of degree \(0\), we have
\begin{equation}
  \mu_2(x_1, \mu_2(x_2, x_3)) + \mu_2(x_3, \mu_2(x_1, x_2)) +\mu_2(x_2, \mu_2(x_3, x_1)) = \mu_1(\mu_3(x_1, x_2, x_3)) +\cdots.
\end{equation}Thus, the standard Jacobi identity for $\mu_2$ may fail up to a homotopy given by the ternary bracket $\mu_3$; when $\mu_3$ (or $\mu_1$) is identically zero, the Jacobi identity for $\mu_2$ holds exactly.

In particular, the higher Jacobi identities imply that $\mu_1$ is a differential, i.e.~$\mu_1(\mu_1(x))=0$, and a graded derivation with respect to $\mu_2$, i.e.\footnote{When no confusion may arise, an $x$ appearing in exponent of $(-1)$ will denote the degree of a homogeneous element $x\in \bigoplus_n V^n$.}
\begin{equation}
  \mu_1(\mu_2(x,y))=\mu_2(\mu_1(x),y)+(-1)^{x}\mu_2(x,\mu_1(y)).
\end{equation}

A \emph{cyclic} $L_\infty$-algebra further carries a \emph{cyclic structure}, an inner product $\langle -,-\rangle$ of degree \(-3\) on $V$, which obeys a graded cyclicity condition with respect to the $\mu_k$, generalising the invariance and cyclicity of the Cartan--Killing form on a conventional Lie algebra. Given a cyclic structure, one can define the \emph{homotopy} Maurer--Cartan (hMC) action\footnote{Technically, either one should insert some signs into \eqref{eq:hMC_action} coming from décalage (see e.g.~\cite{Mehta:2012ppa}), or one should eliminate the signs via the superfield trick \cite[(2.48)]{Jurco:2018sby}. For simplicity we omit this sign.}
\begin{equation}\label{eq:hMC_action}
  S_{\text{hMC}}[x] = \sum_k \frac{1}{(k+1)!} \langle x, \mu_k(x,x,\ldots x)\rangle.
\end{equation}

Now, when the cyclic $L_\infty$-algebra corresponds to that of a BV theory with $\Phi = c+\varphi+\varphi^++c^+$, we have
\begin{equation}
  S_{\text{hMC}}[\Phi] = S_{\text{BV}}[\varphi, c, \varphi^+, c^+],
\end{equation}
making the $L_\infty$/BV correspondence manifest.

In particular, the differential $\mu_1$ provides the kinetic terms
\begin{equation}
  \langle \Phi,\mu_1(\Phi)\rangle = \sum\int\varphi D\varphi+\dotsb
\end{equation}
for $D$ the relevant differential operator for the fields $\varphi$. For example, in the case of Yang--Mills theory, $\mu_1(A)= \mathrm{d} \star \mathrm{d} A$ and $\langle A, A^+\rangle = \mathrm{tr} \int A \wedge A^+$ so that\footnote{In passing from the homotopy Maurer--Cartan form $\langle A, \mu_1(A)\rangle = \mathrm{tr} \int A \wedge \mathrm{d} \star \mathrm{d} A$ to the familiar Yang--Mills kinetic term $\mathrm{tr} \int  \mathrm{d} A \wedge \star \mathrm{d} A$, we have integrated by parts assuming $M$ has no boundary. This story is modified in the presence of boundaries, as discussed in \cite{cricket}. See also \cite{Cattaneo_2014, Cattaneo:2015vsa, Cattaneo:2019jpn, Mnev_2019, Chiaffrino:2023wxk} for important prior  work, from a slightly different perspective, on the BV formalism on manifolds with boundary.}
\begin{equation}
  \langle A, \mu_1(A)\rangle = \mathrm{tr} \int A \wedge \mathrm{d} \star \mathrm{d} A = \mathrm{tr} \int\mathrm{d} A \wedge \star \mathrm{d} A ,
\end{equation}
the usual linearised Yang--Mills kinetic term.
For the complete details, see \cite{Jurco:2018sby}. The $(k+1)$-point interactions are then given by the higher brackets $\mu_k$. For example
\begin{equation}
  \langle A, \mu_2(A, A)\rangle = \mathrm{tr} \int A \wedge \mathrm{d} \star [A,A],
\end{equation}
where
\begin{equation}
  \mu_2\left(A_1, A_2\right)\coloneqq\mathrm{d} \star\left[A_1, A_2\right]+\left[A_1, \star \mathrm{d} A_2\right]+\left[A_2, \star \mathrm{d} A_1\right].
\end{equation}
Said another way, the higher $\mu_k$ determine the Feynman vertices.

It is a nontrivial, but crucial, fact that the $\mu_k$ corresponding to the BV action do indeed form an $L_\infty$-algebra. With this insight in hand, many features of quantum field theory are translated into the mathematics of homotopy algebras and vice versa.

\subsection[{\texorpdfstring{\(L_\infty\)}{𝐿∞}-algebras and scattering amplitudes}]{$\boldsymbol{L_\infty}$-algebras and scattering amplitudes}\label{algebras_and_scattering_amplitudes}
Couched in these terms, (the generating functional of) the $S$-matrix is naturally and inevitably framed in terms of the \emph{minimal model theorem} of homotopy Lie algebras \cite{Kajiura:2001ng,Kajiura:2003ax,Nutzi:2018vkl,Macrelli:2019afx,Arvanitakis:2019ald,Jurco:2019yfd}.

Since $\mu_1$ is a differential, we can take its cohomology\footnote{The equivalence classes of $V$ that are $\mu_1$-closed, $\mu_1(x)=0$, modulo $\mu_1$-exact elements, $x=\mu_1(y)$.}
$\operatorname H^\bullet(V)$.
Recall that the kinetic term of the associated BV action is given by $\mu_1$. Thus $\mu_1(\varphi)=0$ is nothing but the linearised equation of motion for $\varphi$, and the cohomology is the space of physical on-shell asymptotic states (when the \(S\)-matrix is well-defined).\footnote{There are some technical qualifications to this relating to on-shell ghosts or antifields, confinement, etc.; for a related discussion see \cref{sec:stabilityaxiom}.}
The minimal model theorem states that every $L_\infty$-algebra $\mathfrak{L}$, with graded vector space $V$ and $k$-linear products $\mu_k$, is \emph{quasi-isomorphic}\footnote{A quasi-isomorphism is a morphism (a map between $L_\infty$-algebras respecting their higher algebraic structure, see \cref{conventions} and e.g.~\cite{Jurco:2018sby} for more details) that induce an isomorphism between the cohomologies of the domain and the codomain.
That is, it preserves the space of asymptotic on-shell states but allows for field redefinitions of off-shell fields.} to a  minimal $L_\infty$-algebra, $\mathfrak{L}'$, with graded vector space $V'\cong\operatorname H^\bullet(V)$ and $k$-linear products $\mu'_k$ such that $\mu_1'(x')=0$ for all $x'\in V'$.
The minimal model of an \(L_\infty\)-algebra is unique up to morphisms called \(L_\infty\)-isomorphisms.\footnote{See \cref{conventions} and e.g.~\cite{Jurco:2018sby} for details.}
Furthermore, the products \(\mu'_k\) of the minimal model can be constructed (via  \emph{homotopy transfer}) from the $\mu_k$ of the original $L_\infty$-algebra. This procedure generates a sum over diagrams reminiscent of the familiar Feynman diagrams, suggesting that the higher products \(\mu'_k\), or rather their contractions,
\begin{equation}\label{eq:amplitude}
  \langle\varphi_0,\mu'_k(\varphi_1,\dotsc,\varphi_k)\rangle,
\end{equation}
with the cyclic structure, should be interpreted as the \((n=k+1)\)-point components of the tree-level \(S\)-matrix.

This putative interpretation meets the immediate objection that the specific numbers \eqref{eq:amplitude} are not uniquely defined. 
Indeed, a strict isomorphism class of putative $S$-matrices  is given by the \emph{choice} of what is known as the deformation-retract datum \((i,p,h)\) used to produce the minimal model, as briefly reviewed in \cref{Smatrixdeformationretract}. Thus, extra conditions  are required to fix the minimal model encoding the physical $S$-matrix elements \eqref{eq:amplitude} \cite{Kajiura:2003ax, Macrelli:2019afx}. In physically relevant examples, when \(h\) is chosen to agree with a conventional propagator, this then fixes a canonical  deformation-retract datum \((i,p,h)\). These data determine a canonical minimal model, which we shall denote $\fL^\cm=(V^\cm, \mu_k^\cm)$, encoding the physical $S$-matrix \cite{Macrelli:2019afx}. The $L_\infty$-algebra structure on the canonical minimal model $\fL^\cm$ translates directly into the familiar Feynman-diagrammatic construction of the tree-level $S$-matrix \cite{Macrelli:2019afx, Saemann:2020oyz}.

For $\varphi_i^\cm \in V^\cm \cong\operatorname H^\bullet(V)$ (which, recall, is the space of physical on-shell states of the BV theory, e.g.~transverse on-shell gluons in Yang--Mills theory), the $(k+1)$-point tree-level amplitudes $A_{k+1}^{\text{tree-level}}(\varphi_i^\cm)$ with external states specified by $\varphi_i^\cm$ are given by
\begin{equation}
  A_{k+1}^{\text{tree-level}}(\varphi_i^\cm)= \langle \varphi_0^\cm, \mu_k^\cm(\varphi_1^\cm,\varphi_2^\cm,\dotsc,\varphi_{k}^\cm)\rangle.
\end{equation}
More generally, the generating functional of the connected tree-level \(S\)-matrix (excluding the two-point function) is the homotopy Maurer--Cartan action of the minimal model
\begin{equation}
  S^\cm _{\text{hMC}}[\varphi^\cm ] = \sum_k \frac{1}{(k+1)!}\langle \varphi^\cm , \mu_k^\cm (\varphi^\cm ,\varphi^\cm ,\ldots, \varphi^\cm )\rangle.
\end{equation}
The generalisation to loops employs  quantum \(L_\infty\)-algebra \(\fL =( V,\mu_k, \Delta)\) \cite{Zwiebach:1992ie,Markl:1997bj,Pulmann:2016aa,Doubek:2017naz, Jurco:2019yfd,Saemann:2020oyz, jurčo2024lagrangian}, where $\Delta$ is the BV Laplacian. In physically relevant examples, the (formal) loop-level $S$-matrix again follows from homotopy transfer perturbed by $\hbar \Delta$, as in   \cite{Pulmann:2016aa,Doubek:2017naz, Jurco:2019yfd,Saemann:2020oyz}.

To summarise,  actions and amplitudes are unified through quasi-isomorphic cyclic (quantum) $L_\infty$-algebras.
This perspective yields a correspondence between homotopy algebras and perturbative quantum field theories, where various properties of the latter are manifested in (higher) algebraic structure of the former. See \cite{Nutzi:2018vkl, Nutzi:2019ufl, Reiterer:2019dys, Jurco:2018sby, Jurco:2019bvp, Jurco:2019yfd, Macrelli:2019afx, Jurco:2020yyu, Borsten:2020zgj, Saemann:2020oyz, Borsten:2021hua, Borsten:2021gyl, Borsten:2022ouu, Bonezzi:2022yuh, Borsten:2022vtg, Szabo:2022edp, Bonezzi:2022bse, Escudero:2022zdz,  Borsten:2023ned, Borsten:2023reb, Szabo:2023cmv, Bonezzi:2023pox, Bonezzi:2023ced, Bonezzi:2023ciu, Armstrong-Williams:2024icu} for examples ranging from scattering amplitudes recursion relations to colour--kinematics duality.\footnote{There are also close connections to the factorisation algebra programme of Costello and Gwilliam~\cite{Costello:2011aa, Costello:2016vjw,Costello:2021jvx}.}

\subsection{Homotopy algebras and generalised asymptotic observables}\label{algebras_generalised_observables}
The above correspondence raises the obvious question \cite{Macrelli:2019afx}: what to make of those minimal models obtained by using other, noncanonical deformation-retract data \((i,p,h)\)? We provide an interpretation to this question: for appropriate choices of \((i,p,h)\), the resulting minimal models are the Schwinger--Keldysh amplitudes, as derived via LSZ reduction in \cite{Caron-Huot:2023vxl}, which underlie the generalised asymptotic observables. Moreover, an alternative (but equivalent) homotopy transfer yields the generalised cut amplitudes in the perturbative blobology formalism of \cite{Caron-Huot:2023vxl}. Thus, the equivalence between the Schwinger--Keldysh and blobological constructions of generalised cut amplitudes in \cite{Caron-Huot:2023vxl} is translated into the language of quasi-isomorphism between minimal models. 

This has immediate ramifications. It is already surprising that Schwinger--Keldysh amplitudes naturally package themselves into the structure of an \(L_\infty\)-algebra just like scattering amplitudes and that they naturally appear in the homotopy algebra--QFT correspondence. But what's more, since minimal models are unique up to (non-strict) \(L_\infty\)-isomorphisms, it turns out that the \(L_\infty\)-algebra of Schwinger--Keldysh amplitudes is quasi-isomorphic to (an \(N\)-fold direct sum of, where \(N\) is the number of time folds) that encoding the ordinary amplitudes! In particular, this quasi-isomorphism entails recursive formulae that compute Schwinger--Keldysh amplitudes in terms of ordinary amplitudes or vice versa. In one direction, the existence of such a recursive formula is perhaps not surprising --- after all, Schwinger--Keldysh amplitudes are obtained by gluing together ordinary scattering amplitudes along time folds. On the other hand, there also exists a recursive formula in the other direction. That is, Schwinger--Keldysh amplitudes and ordinary scattering amplitudes contain precisely the same physical information in a certain precise sense.

Fitting the Schwinger--Keldysh amplitudes into the minimal model framework summarised in \cref{algebras_and_scattering_amplitudes} presents some obvious difficulties: the cut propagator \(2\pi\mathrm i\delta(p^2-m^2)\Theta(p^0)\), involving such singular objects as the Dirac delta, is not the inverse of a differential operator. Instead, we leverage the freedom in homotopy transfer to choose deformation-retract data that do not necessarily come from the usual propagator; this freedom allows only certain kinds of exotic propagators, but the class of allowed propagators includes the cut propagator appearing in the Schwinger--Keldysh formalism.

New formalisms suggest new perspectives and avenues for generalisation. In the present case, generalising to arbitrary deformation-retract data \((i,p,h)\) delineates a natural class of generalisations of Schwinger--Keldysh amplitudes, which we call \emph{generalised cut amplitudes}, that enjoy the same properties as Schwinger--Keldysh amplitudes, namely quasi-isomorphism to (and thus recursive expressibility in terms of) ordinary amplitudes. (For physically relevant examples, specific generalised cuts coincide with common examples in the scattering amplitudes literature \cite{britto2024cuttingedge}.) Given this observation, it may also be interesting to explore  possible connections between homotopy algebras and the crossing between (generalised) amplitudes developed in \cite{Caron-Huot:2023ikn}. Furthermore, the homotopy-algebraic formalism automatically extends to other types of homotopy algebras. For example, to obtain colour-stripped planar amplitudes in theories with colour or flavour symmetry, one can easily swap out \(L_\infty\)-algebras for \(A_\infty\)-algebras with exactly the same formal properties.

\subsection{Notation and conventions}\label{conventions}
\paragraph{Differential graded vector spaces}
We work over the real numbers \(\mathbb R\). Graded vector spaces are $\mathbb{Z}$-graded: \(V = \bigoplus_{n\in\mathbb Z} V^n\), where each $V^n$ is a vector space of degree $n$. We denote the graded-anticommutative (or graded-antisymmetric) tensor product by \(V \wedge V\), where
\begin{equation}
  x\wedge y = x\otimes y-(-1)^{xy}y\otimes x.
\end{equation}
A differential graded (dg) vector space, or a cochain complex, is a graded vector space $V$ together with a square zero endomorphism $d\in \text{End}^1(V)$ of degree 1. The condition $d^2=0$ allows one to define cohomology
\begin{equation}
  \label{eq:cohomology}
  \operatorname H_d^\bullet(V):=\frac{\text{ker}(d)}{\text{im}(d)}.
\end{equation}
A cochain map, or a map  of dg vector spaces $f:(V,d_V)\to (W,d_W)$ is a linear map such that $d_W\circ f=f\circ d_V$.

Let $[n]$ be the one-dimensional graded vector space concentrated in degree $-n$. For any vector space $V$ we define $V[n]\coloneqq[n]\otimes V$. There is then a natural identification $(V[n])^k\cong V^{n+k}$. We define the suspension map $s\colon V\to V[1]$ by tensoring from the left.

Throughout we work with the Koszul sign convention, where moving symbols past each other results in a sign flip if the symbols both carry odd degrees.

\paragraph{$\boldsymbol{L_\infty}$-algebras}
An $L_\infty$-algebra $\mathfrak{L}$ consists of a graded vector space $V$ and a set of multilinear totally graded-antisymmetric $k$-ary maps
\begin{equation}\mu_k\colon V^{\wedge k}\to V\qquad(k\geq1).\end{equation}
We use cohomological conventions for the theory of $L_\infty$-algebras $\mathfrak{L}=(V,\mu_k)$, such that \(\mu_k\) carries degree $2-k$. Explicitly, a totally graded-antisymmetric bilinear product $\mu_2\colon V\wedge V\to V$ satisfies
\begin{equation}
    \mu_2 (x_2, x_1) = -(-1)^{x_1 x_2}\mu_2 (x_1, x_2),
\end{equation}
with the evident generalisation to $\mu_k $ for arbitrary $k$.

The maps $\mu_k$ must obey the higher, or homotopy, Jacobi identities, explicitly given by
\begin{equation}\label{hjr}
  \sum_{\substack{j+k=i\\\sigma \in \operatorname{Sh}(j ; i)}} \chi\left(\sigma ;  x_1, \ldots,  x_i\right)(-1)^k \mu_{k+1}\left(\mu_j\left( x_{\sigma(1)}, \ldots,  x_{\sigma(j)}\right),  x_{\sigma(j+1)}, \ldots,  x_{\sigma(i)}\right)=0
\end{equation}
for all $i \in \mathbb{N}$ and $ x_1, \ldots,  x_i \in V$, where $\chi\left(\sigma ;  x_1, \ldots,  x_i\right)$ is the graded-antisymmetric Koszul sign\footnote{Related to the Koszul sign defined for the symmetric monoidal category of graded vector spaces by an overall $\operatorname{sgn}(\sigma)$.} defined by
\begin{equation}\label{eq:koszul-sign}
  x_1 \wedge \ldots \wedge  x_i=\chi\left(\sigma;x_1,\ldots,x_i\right)x_{\sigma(1)} \wedge \ldots \wedge  x_{\sigma(i)}.
\end{equation}
The sum is over all shuffles $\sigma \in \operatorname{Sh}(j ; i)$, where the notation \(\operatorname{Sh}(k_1,\dotsc,k_{j-1};i)\) denotes the set of \((k_1,\dotsc,k_{j-1};i)\)-shuffles, i.e.~permutations $\sigma\colon\{1,\dotsc,i\}\to\{1,\dotsc,i\}$ such that $\sigma(1)<\cdots<\sigma(k_1)$, and $\sigma(k_1+1)<\cdots<\sigma(k_1+k_2)$, and so on until $\sigma(k_1+\dotsb+k_{j-1}+1)<\cdots<\sigma(i)$.

In particular, \eqref{hjr} implies $\mu_1\circ \mu_1=0$ so that $(V, \mu_1)$ is a  dg vector space, or a cochain complex, with cohomology $\operatorname H_{\mu_1}^\bullet(V)$, which we will write as $\operatorname H^\bullet(V)$ when no confusion may arise. 

\paragraph{Cyclic $\boldsymbol{L_{\infty}}$-algebras}

A cyclic structure \cite{Kontsevich:1992aa, Penkava:9512014} on an $L_{\infty}$-algebra $\mathfrak{L}=(V, \mu_k)$ is a graded-symmetric nondegenerate bilinear form
\begin{equation}
  \langle-,-\rangle\colon V \times V \rightarrow \mathbb{R}
\end{equation}
which obeys the natural graded cyclicity condition for homogeneous $ x_1, \ldots, x_{i+1} \in V$:\footnote{The sign, apart from the overall \((-1)^i\) (the parity of the permutation), may be formally thought of as the Koszul sign arising from moving the \(x_{i+1}\) past `\(x_1\mu_i(x_2,\dotsc,x_i,-)\)' and then moving \(x_1\) past `\(\mu_i\)'.}
\begin{equation}
  \left\langle x_1, \mu_i\left( x_2, \ldots,  x_{i+1}\right)\right\rangle=(-1)^{i+(2-i)|x_1|+ x_{i+1}(2-i+\sum_{j=1}^i|x_j|)} \left\langle x_{i+1}, \mu_i\left( x_1, \ldots,  x_i\right)\right\rangle.
\end{equation}
A \emph{cyclic} $L_\infty$-algebra $\fL=(V,\mu_k,\langle-,-\rangle)$ is an $L_\infty$-algebra $(V,\mu_k)$ with a cyclic structure $\langle-,-\rangle$.
\paragraph{Quasi-isomorphisms}
A morphism $\phi\colon\fL\to\fL'$ of $L_\infty$-algebras is a family of homogeneous maps $\phi_i\colon V \times \cdots \times V \rightarrow V^{\prime}$ of degree $1-i$ for $i \in \mathbb{N}$ which are multilinear, totally graded-antisymmetric, and obey
\begin{multline}
  \label{eq:coherencemorphisms}
      \sum_{\substack{j+k=i\\\sigma \in \operatorname{Sh}(j ; i)}}(-1)^k \chi\mleft(\sigma ;  x_1, \ldots,  x_i\mright) \phi_{k+1}\mleft(\mu_j\mleft( x_{\sigma(1)}, \ldots,  x_{\sigma(j)}\mright),  x_{\sigma(j+1)}, \ldots,  x_{\sigma(i)}\mright) \\
      =\sum_{j=1}^i \frac{1}{j !} \sum_{\substack{k_1+\cdots+k_j=i\\\sigma \in \operatorname{Sh}\mleft(k_1, \ldots, k_{j-1} ; i\mright)}} \chi\mleft(\sigma ;  x_1, \ldots,  x_i\mright) \zeta\mleft(\sigma ;  x_1, \ldots,  x_i\mright)\hskip5em \\
      \times \mu_j^{\prime}\mleft(\phi_{k_1}\mleft( x_{\sigma(1)}, \ldots,  x_{\sigma\mleft(k_1\mright)}\mright), \ldots, \phi_{k_j}\mleft( x_{\sigma\mleft(k_1+\cdots+k_{j-1}+1\mright)}, \ldots,  x_{\sigma(i)}\mright)\mright),
\end{multline}
where \(\operatorname{Sh}\mleft(k_1, \ldots, k_{j-1} ; i\mright)\) denotes the set of \((k_1,\dotsc,k_{j-1};i)\)-shuffles (defined below \eqref{eq:koszul-sign}) and
where, in addition to the Koszul sign \(\chi\) \eqref{eq:koszul-sign}, we also have the additional sign factor\footnote{This sign factor arises from the general theory of homotopy algebras \cite{Loday:2012aa,Jurco:2018sby}; in particular, with these signs, a morphism of \(L_\infty\)-algebras is equivalent (at least in the finite-dimensional case) to a morphism of free commutative algebras equipped with differentials.}
\begin{equation}
    \zeta\mleft(\sigma ;  x_1, \ldots,  x_i\mright)\coloneqq(-1)^{\sum_{1 \leqslant m<n \leqslant j} k_m k_n+\sum_{m=1}^{j-1} k_m(j-m)+\sum_{m=2}^j\left(1-k_m\right) \sum_{k=1}^{k_1+\cdots+k_{m-1}}x_{\sigma(k)}} .
  \end{equation}
  A morphism of $L_\infty$-algebras is an \emph{$L_\infty$-isomorphism} if the first component map $\phi_1$ is invertible. The coherence relations \eqref{eq:coherencemorphisms} implies in particular that $\phi_1$ is a chain map ($\phi_1\circ \mu_1 = \mu_1'\circ \phi_1$) and thus induces a  map on the cohomology, $\tilde\phi_1\colon\operatorname H^\bullet(V) \to\operatorname H^\bullet(V')$. 
A \emph{quasi-isomorphism} is a morphism such that $\tilde\phi_1$ is an isomorphism. 
A (quasi-iso)morphism $\phi\colon\fL\to\fL'$ between cyclic $L_{\infty}$-algebras is called \emph{cyclic} if it preserves the cyclic structure:
\begin{subequations}\label{cyclicqi}
\begin{equation}
  \left\langle\phi_1\mleft( x_1\mright), \phi_1\mleft( x_2\mright)\right\rangle'=\left\langle x_1,  x_2\right\rangle
\end{equation}
\begin{equation}
  \sum_{\substack{j+k=i \\ j, k \geqslant 1}}\langle\phi_j\mleft( x_1, \ldots,  x_j\mright), \phi_k\mleft( x_{j+1}, \ldots,  x_{j+k}\mright)\rangle'=0 \qquad \forall i \geqslant 3.
\end{equation}
\end{subequations}

\paragraph{Strict morphisms}
\label{sec:strict}
There is a subclass of morphisms of $L_\infty$-algebras called \emph{strict} morphisms. If $\phi\colon\fL\to\fL'$ is a morphism of $L_\infty$-algebras, with component maps $\phi_i\colon\fL^{\wedge i}\to\fL'$, the morphism $\phi$ is called strict if $\phi_i=0$ for $i>1$. Strict quasi-isomorphisms and strict isomorphisms of $L_\infty$-algebras are defined in the obvious way.
\section{Homotopy algebras and the \texorpdfstring{\(\boldsymbol S\)}{𝑆}-matrix}

As briefly reviewed in \cref{algebras_and_scattering_amplitudes}, the tree-level  $S$-matrix of a quantum field theory described by the cyclic $L_\infty$-algebra $\fL=(V,\mu_k)$ is generated by the homotopy Maurer--Cartan action of the minimal model $\fL^\cm$ using an appropriate deformation-retract datum \((i,p,h)\) \cite{Macrelli:2019afx}. This connection to the $S$-matrix is made explicit via \emph{homotopy transfer}, which expresses the products $\mu^\cm_k$ of the minimal model $\fL^\cm$ through the Feynman diagrams constructed from the vertices $\mu_k$ of $\fL$.

\subsection{The minimal model theorem}
\label{Smatrixdeformationretract}
First, we have the deformation retract. Restricting to the free theory (i.e.~forgetting about all higher products $\mu_{k}$ for $k>1$), $(V, \mu_1)$ and $(\operatorname H^\bullet(V), \mu'_1\coloneqq 0)$ form cochain complexes
since $\mu_1\circ\mu_1 =0$ by the homotopy Jacobi relations \eqref{hjr}. Then there exists a pair of cochain maps,
$p\colon V\to\operatorname H^\bullet(V)$ and $i\colon\operatorname H^\bullet(V)\to V$, both of degree \(0\) and such that $p\circ i=\id_{\operatorname H^\bullet(V)}$.  Physically, $p$ projects off-shell field configurations onto the corresponding on-shell states, while $i$ embeds the on-shell states in the space of off-shell fields. Thus, the condition $p\circ i=\id_{\operatorname H^\bullet(V)}$ simply states that if one embeds an on-shell state in the space of off-shell fields and then projects back to the space of on-shell states, the original on-shell state is recovered. On the other hand, $i\circ p\not=\id_{V}$ generically, since many off-shell states correspond to the same on-shell state.

Although $i\circ p\not=\id_{V}$, there always exists a cochain homotopy $h$ relating $i\circ p$ to $\id_V$, i.e.~a degree-$(-1)$ map $h\colon V\to V$, such that
\begin{equation}
  \id_V-i\circ p =\mu_1h+h\mu_1.
\end{equation}
A triple $(i,p,h)$ satisfying this relation is a known as (the datum of) a deformation retract, see e.g.~\cite{Crainic:0403266}. They are usually depicted diagrammatically as
\begin{equation}
  \begin{tikzcd}[ampersand replacement=\&]
    \label{eq:ogretract}
    \ar[loop left,out=150,in=210,distance=30,"h" left,xshift=0pt] (V, \mu_1)\arrow[r,twoheadrightarrow,shift left]{}{p} \& (\operatorname H^\bullet(V), 0)\arrow[l,hookrightarrow,shift left]{}{i}
  \end{tikzcd}.
\end{equation}
It is clear, then, that the homotopy $h$ inverts the differential $\mu_1$, which gives the free equations of motion, up to the on-shell states. In particular, the familiar propagators of scattering theory provide  examples of \(h\). Note that the choice of gauge, when required, is implicit in the choice of $i$. 

Now, using homotopy transfer  one can include the interactions $\mu_k$ on $V$  and transfer them along the retract \eqref{eq:ogretract}, lifting the chain map $i$ to a full quasi-isomorphism $\phi^{(i,p,h)}\colon\fL'\to \fL$.\footnote{We have labelled the quasi-isomorphism with the deformation retract triple $(i,p,h)$, since it depends on the choice of data.}
Explicitly, starting from $\phi^{(i,p,h)}_1(x')\coloneqq i x'$, the multilinear maps $\phi^{(i,p,h)}_k\colon (\operatorname H^\bullet(V))^{\wedge k}\to V$, defined recursively by
\begin{multline}
  \label{eq:injectionmap}
  \phi^{(i,p,h)}_n\mleft(x_1^{\prime}, \ldots, x_n^{\prime}\mright)\coloneqq-\sum_{j=2}^n \frac{1}{j !} \sum_{\substack{k_1+\cdots+k_j=n\\\sigma \in \operatorname{Sh}\mleft(k_1, \ldots, k_{j-1} ; n\mright)}}\chi\mleft(\sigma ; x_1^{\prime}, \ldots, x_n^{\prime}\mright) \zeta\mleft(\sigma ; x_1^{\prime}, \ldots, x_n^{\prime}\mright)\\
  \times h\mu_j\mleft(\phi^{(i,p,h)}_{k_1}\mleft(x_{\sigma(1)}^{\prime}, \ldots, x_{\sigma\left(k_1\right)}^{\prime}\mright), \ldots, \phi^{(i,p,h)}_{k_j}\mleft(x_{\sigma\mleft(k_1+\cdots+k_{j-1}+1\mright)}^{\prime}, \ldots, x_{\sigma(n)}^{\prime}\mright)\mright)
\end{multline}
yield a quasi-isomorphism if and only if the products of $\fL'$ are taken to be:
\begin{multline}\label{murecursive}
  \mu'_n(x'_1, \ldots,  x'_n)\coloneqq\sum_{j=2}^n \frac{1}{j !} \sum_{\substack{k_1+\cdots+k_j=n\\\sigma\in\operatorname{Sh}(k_1, \ldots, k_{j-1} ; n )}}\chi ( \sigma ;  x'_1, \ldots,  x'_n) \zeta ( \sigma ;  x'_1, \ldots,  x'_n) \\
  \times p \mu_j\mleft(\phi^{(i,p,h)}_{k_1}\mleft(x'_{\sigma(1)}, \ldots, x'_{\sigma(k_1)}\mright), \ldots, \phi^{(i,p,h)}_{k_j}\mleft(x'_{\sigma(k_1+\cdots+k_{j-1}+1)}, \ldots,  x'_{\sigma(n)}  \mright)\mright).
\end{multline}
This is known as \emph{homotopy transfer}, i.e.~we pertubatively transfer the higher products which encode the homotopy relaxation of the algebraic identities.
Put another way, the $\mu'_k$ of the minimal model are recursively constructed from the vertices $\mu_k$ and the deformation-retract datum \((i,p,h)\). The expression \eqref{murecursive} for $\mu_n'$ can be depicted pictorially as a sum over all rooted trees with $n$ leaves, where one labels the leaves with $i$, the root with $p$, $h$ on each internal line, and each vertex with appropriate $\mu_k$,

\begin{equation}
  \label{eq:homotopytransfer}
  \mu_n'=\sum_{\text{rooted trees}}\pm
  \begin{tikzpicture}
    [scale=0.5,baseline={([yshift=-1ex]current bounding box.center)}]
    \draw [thick] (0,0) -- (0,-1);
    \draw [thick] (0,0) -- (0.7,0.7);
     \draw [thick] (0,0) -- (-0.7,0.7);
    
    \draw [thick] (1.3,1.3) -- (2,2);

    \draw [thick] (-1.3,1.3) -- (-2,2);
    
    \draw [thick] (2,2) -- (3,3);
    \draw [thick] (2,2) -- (1,3);

    \draw [thick] (0,5.5) -- (0,4.5);
    \draw [thick] (0,5.5) -- (0.7,6.2);
    \draw [thick] (0,5.5) -- (-0.7,6.2);
    
    \node at (1,1) {\footnotesize$h$};
    \node at (-1,1) {\footnotesize$h$};
    \node at (0,4.2) {\footnotesize$h$}; 
    \node at (0,2.5) {$\cdots$};
    \node at (0,3.7) {$\vdots$};
   
    \node at (0,-1)[anchor=north] {\footnotesize\(p\)};
    \node at (-3,6)[anchor=south] {\footnotesize\(i\)};
    \node at (3,6) [anchor=south] {\footnotesize\(i\)};
    \node at (0,6) [anchor=south] {\footnotesize\(i\)};
    \node at (-1.5,6)[anchor=south] {\footnotesize\(\cdots\)};
    \node at (1.5,6)[anchor=south] {\footnotesize\(\cdots\)};

    \fill[gray!50] (0,5.5) circle (0.45cm);
    \draw [thick] (0,5.5) circle (0.45cm);
    \node at (0,5.5) {\footnotesize$\mu_j$};

    \fill[gray!50] (0,0) circle (0.45cm);
    \draw [thick] (0,0) circle (0.45cm);
    \node at (0,0) {\footnotesize$\mu_l$};

    \fill[gray!50] (2,2) circle (0.45cm);
    \draw [thick] (2,2) circle (0.45cm);
    \node at (2,2) {\footnotesize$\mu_k$};

  \end{tikzpicture}.
\end{equation}

For example, up to signs we have at four points
\begin{multline}\label{4points}
    \mu'_3(x'_1, x'_2, x'_3) =  p \mu_2( h \mu_2( i x'_1, i x'_2), i x'_3 ) +  p \mu_2( h \mu_2( i x'_3, i x'_1), i x'_2 )\\
    +  p \mu_2( h \mu_2( i x'_2, i x'_3), i x'_1 )+ p \mu_3( i x'_1, i x'_2, i x'_3 ).
\end{multline}

\subsection[The minimal model of the \texorpdfstring{\(S\)}{𝑆}-matrix]{The minimal model of the \(\boldsymbol S\)-matrix}

To connect the minimal model theorem to scattering amplitudes, let us now assume: 
\begin{enumerate}[(i)]
\item  $\fL$ is  equipped with a cyclic structure, denoted  \(\langle-,-\rangle\), and  the hMC action,
\begin{equation}
  S_{\text{hMC}}[\varphi] = \sum_k \frac{1}{(k+1)!} \langle \varphi, \mu_k(\varphi,\varphi,\dotsc,\varphi)\rangle,
\end{equation}
is the BV action of the corresponding quantum field theory.
\item   The homotopy \(h\) in the deformation-retract datum \((i,p,h)\) corresponds to a conventional   propagator\footnote{For physically relevant examples, this canonical choice yields a Hodge--Kodaira decomposition \cite{Macrelli:2019afx}, which is equivalent to having  \emph{special}  deformation-retract data, i.e.~the side conditions $h^2=ih=hp=0$ hold \cite{Chuang_2009}.} as described  in \cite{Macrelli:2019afx} for various physically relevant examples. 
\end{enumerate}

Then the corresponding canonical minimal model, $\fL^\cm$, has hMC action
\begin{equation}\label{mmhMC}
  S^\cm _{\text{hMC}}[\varphi^\cm ] = \sum_k \frac{1}{(k+1)!}\langle\varphi^\cm , \mu_k^\cm (\varphi^\cm ,\varphi^\cm ,\dotsc,\varphi^\cm )\rangle^\cm, 
\end{equation}
which is the generating functional of the $S$-matrix,
where
\begin{equation}
    \langle x_1,x_2\rangle^\cm\coloneqq \langle ix_1,ix_2\rangle
\end{equation}
is the induced cyclic structure on the minimal model.
The recursive definition of the products $\mu_k^\cm$, given in \eqref{murecursive}, then provides a Feynman-diagrammatic expansion of the amplitudes. For example, plugging \eqref{4points} into the quartic term, we have
\begin{multline}
    \langle \varphi^\cm _4,  \mu^\cm _3(\varphi^\cm _1, \varphi^\cm _2, \varphi^\cm _3)\rangle^\cm 
    = \langle \varphi^\cm _4,   p \mu_2( h \mu_2( i \varphi^\cm _1, i \varphi^\cm _2), i \varphi^\cm _3 )\rangle+\langle \varphi^\cm _4,  p \mu_2( h \mu_2( i \varphi^\cm _3, i \varphi^\cm _1), i \varphi^\cm _2 )\rangle\\
    +\langle \varphi^\cm _4,  p \mu_2( h \mu_2( i \varphi^\cm _2, i \varphi^\cm _3), i \varphi^\cm _1 ) \rangle+\langle \varphi^\cm _4, p \mu_3( i \varphi^\cm _1, i \varphi^\cm _2, i \varphi^\cm _3 )\rangle^\cm ,
\end{multline}
which we recognise (recalling that $h$ is the propagator on internal lines) as the familiar sum of the $s$-, $t$- and $u$-channels with the four-point contact term (where we have left the final pairing with the outstate $\phi^\cm_4$ implicit)
\begin{equation}
  \begin{tikzpicture}[scale=0.7,baseline={([yshift=-1ex]current bounding box.center)}]
    \draw [thick] (-2,0) -- (-0.6,-1.4);
    \draw [thick] (-0.35,-1.65) -- (0,-2);
    \draw [thick] (2,0) -- (0,-2);
    \draw [thick] (0,-2) -- (0,-3);
    \draw [thick] (-1,-1)--(0,0);
    \node at (-.5,-1.5) {\footnotesize{$h$}};
    \node at (0,-3.5) {\footnotesize\(p\)};
    \node at (0,0.5) {\footnotesize\(i\)};
    \node at (-2,0.5) {\footnotesize\(i\)};
    \node at (2,0.5) {\footnotesize\(i\)};
    
     \fill[gray!50] (-1,-1) circle (0.25cm);
    \draw [thick] (-1,-1) circle (0.25cm);
    \node at (-0.98,-1) {\footnotesize$\mu_2$};
    \fill[gray!50] (0,-2) circle (0.25cm);
    \draw [thick] (0,-2) circle (0.25cm);
    \node at (0.02,-2) {\footnotesize$\mu_2$};
  \end{tikzpicture}
  +
  \begin{tikzpicture}[scale=0.7,baseline={([yshift=-1ex]current bounding box.center)}]
    \draw [thick] (-2,0) -- (0,-2);
    \draw [thick] (2,0) -- (0.65,-1.35);
    \draw [thick] (0,-2) -- (0.4,-1.6);
    \draw [thick] (0,-2) -- (0,-3);
    \draw [thick] (1,-1)--(0,0);
    \node at (0.55,-1.45) {\footnotesize{$h$}};
     \fill[gray!50] (1,-1) circle (0.25cm);
    \draw [thick] (1,-1) circle (0.25cm);
    \node at (1.02,-1) {\footnotesize$\mu_2$};
    \fill[gray!50] (0,-2) circle (0.25cm);
    \draw [thick] (0,-2) circle (0.25cm);
    \node at (0.02,-2) {\footnotesize$\mu_2$};\node at (0,-3.5) {\footnotesize\(p\)};
    \node at (0,0.5) {\footnotesize\(i\)};
    \node at (-2,0.5) {\footnotesize\(i\)};
    \node at (2,0.5) {\footnotesize\(i\)};
  \end{tikzpicture}
  +
  \begin{tikzpicture}[scale=0.7,baseline={([yshift=-1ex]current bounding box.center)}]
    \draw [thick] (-2,0) -- (-0.6,-1.4);
    \draw [thick] (-0.35,-1.65) -- (0,-2);
    
    \draw [thick] (1,-1) -- (0,-2);
    \draw [thick] (0,-3) -- (0,-2);
    \draw [thick] (1,-1)--(0,0);
    \draw [thick] (-1,-1) -- (0.3,-0.57);
    \draw [thick] (2,0) -- (0.7,-.45);
    \node at (-.5,-1.5) {\footnotesize{$h$}};
  
    \node at (0,-3.5) {\footnotesize\(p\)};
    \node at (0,0.5) {\footnotesize\(i\)};
    \node at (-2,0.5) {\footnotesize\(i\)};
    \node at (2,0.5) {\footnotesize\(i\)};

     \fill[gray!50] (-1,-1) circle (0.25cm);
    \draw [thick] (-1,-1) circle (0.25cm);
    \node at (-0.98,-1) {\footnotesize$\mu_2$};
    \fill[gray!50] (0,-2) circle (0.25cm);
    \draw [thick] (0,-2) circle (0.25cm);
    \node at (0.02,-2) {\footnotesize$\mu_2$};
  \end{tikzpicture}
  +
  \begin{tikzpicture}[scale=0.7,baseline={([yshift=-1ex]current bounding box.center)}]
    \draw [thick] (-2,0) -- (0,-2);
    \draw [thick] (2,0) -- (0,-2);
    \draw [thick] (0,0) -- (0,-2);
    \draw [thick] (0,-2) -- (0,-3);
    \node at (0,-3.5) {\footnotesize\(p\)};
    \node at (0,0.5) {\footnotesize\(i\)};
    \node at (-2,0.5) {\footnotesize\(i\)};
    \node at (2,0.5) {\footnotesize\(i\)};
    \fill[gray!50] (0,-2) circle (0.25cm);
    \draw [thick] (0,-2) circle (0.25cm);
    \node at (0.02,-2) {\footnotesize$\mu_3$};
  \end{tikzpicture}
\end{equation}
of the four-point tree-level amplitude, i.e.
\begin{equation}
  \langle \varphi^\cm _4, \mu^\cm _3(\varphi^\cm _1, \varphi^\cm _2, \varphi^\cm _3)\rangle^\cm = A_{4}^{\text{tree-level}}(\varphi^\cm _1, \varphi^\cm _2, \varphi^\cm _3, \varphi^\cm _4).
\end{equation}
It is clear that this generalises to arbitrary tree-level amplitudes, see e.g.~\cite{Saemann:2020oyz} and references therein. For the quantum theory \cite{Jurco:2019yfd,Saemann:2020oyz}, loop-level scattering amplitudes appear as the corresponding minimal model for a loop or quantum \(L_\infty\)-algebra \cite{Pulmann:2016aa,Doubek:2017naz}. 

In this context, quasi-isomorphisms between $L_\infty$-algebras correspond to the  physical equivalence of the associated theories, e.g.~given by field redefinitions,  in the following sense. Consider a cyclic $L_\infty$-algebra  associated to a BV-action with a canonical minimal model associated to the corresponding tree-level $S$-matrix. A second, quasi-isomorphic, cyclic  $L_\infty$-algebra  will have a representative in its equivalence class of minimal models   that yields an identical $S$-matrix using an appropriate cyclic structure\footnote{This choice  may not coincide with that induced from the original  cyclic structure if the quasi-isomorphism is not cyclic \eqref{cyclicqi}.}.  In this sense it encodes the same physical content, but in general this representative need \emph{not} correspond identically to its tree-level $S$-matrix, as given by its canonical minimal model. 

This possibility follows from the fact that a minimal model $\fL'$ of a homotopy Lie algebra $\fL$ is merely unique up to $L_\infty$-isomorphism. In particular, the minimal model will depend in general on the choice of the deformation retract \((i,p,h)\). When the deformation-retract datum \((i,p,h)\) is given by the physical propagator, homotopy transfer produces the \(S\)-matrix. However, for a generic choice of deformation-retract datum $(i,p,h)$, the homotopy \(h\) need not be the propagator. In this case, the corresponding minimal model will merely be quasi-isomorphic to that of the $S$-matrix, not identical to it.

\section{Homotopy algebras and generalised cut amplitudes}
Starting from an $L_\infty$-algebra with homotopy Maurer--Cartan action $S_{\text{hMC}}[\varphi]$ associated to a BV action $S_{\text{BV}}[a]$, the $S$-matrix is then generated by the homotopy Maurer--Cartan action $S^\cm_{\text{hMC}}[\varphi^\cm]$ of the unique (up to strict isomorphism) minimal model given by the canonical cyclic special deformation retract, i.e.~one that agrees with the propagator. Here we introduce generalised cut amplitudes through a deformation of the canonical deformation-retract datum. In particular, by deforming the homotopy (i.e.\ propagator) $h\mapsto h+b$, we obtain a quasi-isomorphic minimal model whose hMC action is the generating functional of the generalised cut amplitudes.

\subsection{Generalised cut amplitudes}
Suppose $\fL=(V, \mu_k, \langle-,-\rangle)$ is the cyclic $L_\infty$-algebra, corresponding to some BV-theory, and that \((i,p,h)\) is the canonical deformation-retract datum between $\fL$ and the minimal model $\fL^\cm=( H^\bullet( V), \mu^\cm_k{}^{(i,p,h)}, \langle-,-\rangle^\cm)$. We have labelled the products $\mu^\cm_k{}^{(i,p,h)}$ to emphasise that they are generated via homotopy transfer by the choice of the deformation-retract datum \((i,p,h)\). When no confusion may arise we simplify the notation to $\mu_k^\cm$. 
\begin{definition}[Generalised cut propagator] Let \((i,p,h)\) be a deformation-retract datum of cochain complexes
  \begin{equation}
    \begin{tikzcd}
      \ar[loop left, "h"] (V, \mu_1)\rar[twoheadrightarrow,shift left,"p"] & (\operatorname H^\bullet(V), 0)\lar[hookrightarrow,shift left,"i"]
    \end{tikzcd}.
  \end{equation}
  A \emph{generalised cut propagator} is a degree $(-1)$ map \(b\colon V\to V\) of the form
  \begin{equation}
    \label{eq:SKcut}
    b\coloneqq i\beta p,
  \end{equation}
  where $\beta\colon\operatorname H^\bullet(V)\to\operatorname H^\bullet(V)$ is an arbitrary linear map of degree minus one.
\end{definition}

Note that the cochain map conditions imply \(\mu_1i=0=p\mu_1\)  so that
\begin{equation}
  \mu_1 b+b\mu_1 =\mu_1 i\beta p+i\beta p\mu_1 =0.
\end{equation}
Thus, \((i,p,h+b)\) also forms a deformation-retract datum, but generically \(h+b\) will not be the ordinary propagator. In particular, generically the side conditions will fail. 

In terms of homotopy transfer and scattering amplitudes, \(b\) represents a propagator that is `cut' in the sense that it involves a projection to the cohomology $\operatorname H^\bullet( V)$. In the context of scattering amplitudes, this corresponds to the usual notion of a cut propagator \cite[\S V.A]{britto2024cuttingedge}, since the cohomology is precisely the space of physical asymptotic on-shell states of the theory. In particular, the cut propagators of Schwinger--Keldysh Feynman diagrams are a special case thereof.

\begin{definition}[\(L_\infty\)-algebra of generalised cut amplitudes]
  Let  $\fL=(V, \mu_k, \langle-,-\rangle)$ be a cyclic $L_\infty$-algebra with the canonical deformation retract $(i,p,h)$ to its minimal model $\fL^\cm=(\operatorname H^\bullet(V), \mu^\cm_k, \langle-,-\rangle^\cm)$. Given a generalised cut propagator $b$,
  the \emph{\(L_\infty\)-algebra of generalised cut amplitudes} is the minimal model \(\tilde{\fL}=(\operatorname H^\bullet( V), \tilde{\mu}_k, \langle-,-\rangle\tilde{}~)\) given by the deformation-retract datum \((i,p,h+b)\).
  The \emph{generalised cut amplitudes} associated to the generalised cut propagator \(b\) are the numbers
  \begin{equation}
    \langle \tilde\varphi_0,\tilde\mu_n(\tilde\varphi_1,\dotsc,\tilde\varphi_n)\rangle\tilde{}
  \end{equation}
  given by the \(L_\infty\)-algebra of generalised cut amplitudes.
\end{definition}

Now, given the construction, the following theorem is immediate.
\begin{theorem}\label{thm:keldysh-schwinger-is-quasi-iso}
  The \(L_\infty\)-algebra of generalised cut amplitudes $\tilde\fL$ is quasi-isomorphic to the \(L_\infty\)-algebra of the $S$-matrix \(\fL^\cm\), and the leading-order component \(\rho_1\) of the quasi-isomorphism \(\rho\colon\tilde\fL\to\fL^\cm\) is a bijection.
\end{theorem}\begin{proof}
  By the properties of homotopy transfer, both are quasi-isomorphic to \(\fL\), and quasi-isomorphism is an equivalence relation.
\end{proof}
Thus, in particular, there exist recursion formulae that relate ordinary amplitudes to generalised cut amplitudes and vice versa.
Note that there also exist quasi-isomorphisms in the other directions, $\fL\to \tilde{\fL}$ and $\fL^\cm\to \fL$, so that we get a span of quasi-isomorphisms of $L_\infty$ algebras:\footnote{Such spans are universal in that any quasi-isomorphism can be so represented \cite{Farahani:2023ptb}.}
\begin{equation}
  \begin{tikzcd}
   & \fL \ar[dr,shift left] \ar[dl,shift left]& \\
    \tilde{\mathfrak L} \ar[ur,shift left] &  & \mathfrak L^\cm \ar[ul,shift left].
  \end{tikzcd}
\end{equation}

\subsection[Recursion relations between the \texorpdfstring{\(S\)}{𝑆}-matrix and generalised cut amplitudes]{Recursion relations between the \(\boldsymbol S\)-matrix and generalised cut amplitudes}\label{recursion}

Let us be more explicit and draw out how the generalised cut amplitudes can be computed in terms of the $S$-matrix.

Given a cyclic $L_\infty$-algebra \(\mathfrak L=(V,\mu_k,\langle-,-\rangle)\) and the canonical deformation retract
\begin{equation}
  \label{eq:canonicaldr}
  \begin{tikzcd}
    (V,\mu_1)\ar[loop left, "h"] \rar[twoheadrightarrow,shift left, "p"] & (V^\cm,0) \lar[hookrightarrow,shift left, "i"],
  \end{tikzcd} \qquad \text{where} \qquad V^\cm\cong \operatorname H^\bullet(V),
\end{equation}
together with a generalised cut propagator \(b=i\beta p\), we also have the (not necessarily canonical) deformation retract
\begin{equation}
  \label{eq:cutdr}
  \begin{tikzcd}
    (V,\mu_1)\ar[loop left, "h+b"] \rar[twoheadrightarrow,shift left, "p"] & (\tilde{V},0) \lar[hookrightarrow,shift left, "i"],
  \end{tikzcd} \qquad \text{where} \qquad \tilde{V}\cong \operatorname H^\bullet(V).
\end{equation}
We can now apply homotopy transfer along these two retracts, as  explained in \cref{Smatrixdeformationretract}. The procedure yields formulae for the $L_\infty$-algebra structure on $V^\cm$ and $\tilde{V}$ in terms of sums over trees:
\begin{subequations}
  \begin{align}
    \label{eq:cmfromog}
    \mu_k^\cm&=\sum_{\hskip-1em\tau\in\{\text{rooted trees}\}\hskip-1em}\pm \tau(i,\mu,h,p),\\
    \label{eq:tildefromog}
    \tilde\mu_k&=\sum_{\hskip-1em\tau\in\{\text{rooted trees}\}\hskip-1em}\pm \tau(i,\mu,h+b,p),
  \end{align}
\end{subequations}
  where the sums range over all possible rooted trees, and the map $\tau(i,\mu,x,p)$ is defined by inserting $i$ on the leaves, $\mu_k$ on each $(k+1)$-ary vertex, and $x$ on each internal edge of the tree $\tau$, with $p$ on the root. This procedure gives the ordinary Feynman diagram expansion of the ordinary amplitudes, and an expansion of the generalised cut amplitudes in terms of the ordinary Feynman diagram expansion, supplemented by identical diagrams with the propagator $h$ replaced by $b$ in all possible ways.
  
  In fact, we can go one step further and express the generalised cut amplitudes in terms of the ordinary amplitudes directly by considering the retract\begin{equation}
  \begin{tikzcd}\label{eq:trivialretract}
    ({V}^\circ\cong\operatorname{H}^\bullet(V),0)\ar[loop left, "\beta"] \rar[shift left, "\id"] & (\tilde V\cong\operatorname{H}^\bullet(V),0) \lar[shift left, "\id"].
  \end{tikzcd}
\end{equation}
Interpreting  $(\tilde{V}\cong \operatorname H^\bullet(V), 0)$ as the cohomology of $(V^\cm\cong \operatorname H^\bullet(V), 0)$, we can perform homotopy transfer along \eqref{eq:trivialretract}, where the higher products on $V^\cm$ are the $\mu_k^\cm$ coming from homotopy transfer along the canonical deformation retract \eqref{eq:canonicaldr}. We can then express the generalised cut amplitudes as
 \begin{equation}
   \label{eq:tildefromcm}
    \tilde\mu_k=\sum_{\hskip-1em\tau\in\{\text{rooted trees}\}\hskip-1em}\pm \tau(\id,\mu^\cm,\beta,\id),
\end{equation}
where the vertices of the trees are now replaced with the products on $\fL^\cm$, the internal edges with  $\beta$, and the leaves and the root with the identity. Thus, the generalised cut amplitudes, encoded in the $\tilde\mu_k$,  are expressed in this manner through the familiar $S$-matrix elements, encoded in the $\mu^\cm_k$, associated to the canonical minimal model. Thus, the formulation of the generalised cut amplitudes in terms of \eqref{eq:tildefromcm} corresponds precisely to the blobology of \cite{Caron-Huot:2023vxl}, where the generalised cut amplitudes are constructed by gluing together ordinary amplitudes (the blobs) with cut propagators.

The expansion of \eqref{eq:tildefromcm}, upon further expansion of the $\mu_k^\cm$ using \eqref{eq:cmfromog}, obviously agrees with \eqref{eq:tildefromog} (recall that $b=i\beta p$).
For example, at four points, we have 
\begin{equation}\label{eq:mu3}
  \begin{split}
   \tilde{\mu}_3&=\begin{tikzpicture}[scale=0.6,baseline={([yshift=-1ex]current bounding box.center)}]
    \draw [thick] (-2,0) -- (0,-2);
    \draw [thick] (2,0) -- (0,-2);
    \draw [thick] (0,0) -- (0,-2);
    \draw [thick] (0,-2) -- (0,-3);
    \fill[gray!50] (0,-2) circle (0.4cm);
    \draw [thick] (0,-2) circle (0.4cm);
    \node at (0.02,-2) {\footnotesize{$\mu_3^\cm$}};
  \end{tikzpicture}
     +               
      \begin{tikzpicture}[scale=0.6,baseline={([yshift=-1ex]current bounding box.center)}]
    \draw [thick] (-2,0) -- (-0.6,-1.4);
    \draw [thick] (-0.35,-1.65) -- (0,-2);
    \draw [thick] (2,0) -- (0,-2);
    \draw [thick] (0,-2) -- (0,-3);
    \draw [thick] (-1,-1)--(0,0);
    \node at (-.5,-1.5) {\scriptsize{$\beta$}};
     \fill[gray!50] (-1,-1) circle (0.4cm);
    \draw [thick] (-1,-1) circle (0.4cm);
    \node at (-0.98,-1) {\footnotesize$\mu_2^\cm$};
    \fill[gray!50] (0,-2) circle (0.4cm);
    \draw [thick] (0,-2) circle (0.4cm);
    \node at (0.02,-2) {\footnotesize$\mu_2^\cm$};
  \end{tikzpicture}
  -
  \begin{tikzpicture}[scale=0.6,baseline={([yshift=-1ex]current bounding box.center)}]
    \draw [thick] (-2,0) -- (0,-2);
    \draw [thick] (2,0) -- (0.65,-1.35);
    \draw [thick] (0,-2) -- (0.4,-1.6);
    \draw [thick] (0,-2) -- (0,-3);
    \draw [thick] (1,-1)--(0,0);
    \node at (0.55,-1.45) {\scriptsize{$\beta$}};
     \fill[gray!50] (1,-1) circle (0.4cm);
    \draw [thick] (1,-1) circle (0.4cm);
    \node at (1.02,-1) {\footnotesize$\mu_2^\cm$};
    \fill[gray!50] (0,-2) circle (0.4cm);
    \draw [thick] (0,-2) circle (0.4cm);
    \node at (0.02,-2) {\footnotesize$\mu_2^\cm$};
  \end{tikzpicture}
  -
  \begin{tikzpicture}[scale=0.6,baseline={([yshift=-1ex]current bounding box.center)}]
    \draw [thick] (-2,0) -- (-0.6,-1.4);
    \draw [thick] (-0.35,-1.65) -- (0,-2);    
    \draw [thick] (1,-1) -- (0,-2);
    \draw [thick] (0,-3) -- (0,-2);
    \draw [thick] (1,-1)--(0,0);
    \draw [thick] (-1,-1) -- (0.3,-0.57);
    \draw [thick] (2,0) -- (0.7,-.45);
    \node at (-.5,-1.5) {\scriptsize{$\beta$}};
     \fill[gray!50] (-1,-1) circle (0.4cm);
    \draw [thick] (-1,-1) circle (0.4cm);
    \node at (-0.98,-1) {\footnotesize$\mu_2^\cm$};
    \fill[gray!50] (0,-2) circle (0.4cm);
    \draw [thick] (0,-2) circle (0.4cm);
    \node at (0.02,-2) {\footnotesize$\mu_2^\cm$};
  \end{tikzpicture}\\
    &=
    \begin{tikzpicture}[scale=0.6,baseline={([yshift=-1ex]current bounding box.center)}]
    \draw [thick] (-2,0) -- (0,-2);
    \draw [thick] (2,0) -- (0,-2);
    \draw [thick] (0,0) -- (0,-2);
    \draw [thick] (0,-2) -- (0,-3);
    \fill[gray!50] (0,-2) circle (0.4cm);
    \draw [thick] (0,-2) circle (0.4cm);
    \node at (0.02,-2) {\footnotesize{$\mu_3$}};
    \node at (0,-3.5) {\footnotesize\(p\)};
    \node at (0,0.5) {\footnotesize\(i\)};
    \node at (-2,0.5) {\footnotesize\(i\)};
    \node at (2,0.5) {\footnotesize\(i\)};
  \end{tikzpicture}
  +
       \begin{tikzpicture}[scale=0.6,baseline={([yshift=-1ex]current bounding box.center)}]
    \draw [thick] (-2,0) -- (-0.6,-1.4);
    \draw [thick] (-0.35,-1.65) -- (0,-2);
    \draw [thick] (2,0) -- (0,-2);
    \draw [thick] (0,-2) -- (0,-3);
    \draw [thick] (-1,-1)--(0,0);
    \node at (-.5,-1.5) {\scriptsize{$h$}};
    \node at (0,-3.5) {\footnotesize\(p\)};
    \node at (0,0.5) {\footnotesize\(i\)};
    \node at (-2,0.5) {\footnotesize\(i\)};
    \node at (2,0.5) {\footnotesize\(i\)};
     \fill[gray!50] (-1,-1) circle (0.4cm);
    \draw [thick] (-1,-1) circle (0.4cm);
    \node at (-0.98,-1) {\footnotesize$\mu_2$};
    \fill[gray!50] (0,-2) circle (0.4cm);
    \draw [thick] (0,-2) circle (0.4cm);
    \node at (0.02,-2) {\footnotesize$\mu_2$};
  \end{tikzpicture}
  -
  \begin{tikzpicture}[scale=0.6,baseline={([yshift=-1ex]current bounding box.center)}]
    \draw [thick] (-2,0) -- (0,-2);
    \draw [thick] (2,0) -- (0.65,-1.35);
    \draw [thick] (0,-2) -- (0.4,-1.6);
    \draw [thick] (0,-2) -- (0,-3);
    \draw [thick] (1,-1)--(0,0);
    \node at (0.55,-1.45) {\scriptsize{$h$}};
     \fill[gray!50] (1,-1) circle (0.4cm);
    \draw [thick] (1,-1) circle (0.4cm);
    \node at (1.02,-1) {\footnotesize$\mu_2$};
    \fill[gray!50] (0,-2) circle (0.4cm);
    \draw [thick] (0,-2) circle (0.4cm);
    \node at (0.02,-2) {\footnotesize$\mu_2$};\node at (0,-3.5) {\footnotesize\(p\)};
    \node at (0,0.5) {\footnotesize\(i\)};
    \node at (-2,0.5) {\footnotesize\(i\)};
    \node at (2,0.5) {\footnotesize\(i\)};
  \end{tikzpicture}
  -
  \begin{tikzpicture}[scale=0.6,baseline={([yshift=-1ex]current bounding box.center)}]
    \draw [thick] (-2,0) -- (-0.6,-1.4);
    \draw [thick] (-0.35,-1.65) -- (0,-2);
    \draw [thick] (1,-1) -- (0,-2);
    \draw [thick] (0,-3) -- (0,-2);
    \draw [thick] (1,-1)--(0,0);
    \draw [thick] (-1,-1) -- (0.3,-0.57);
    \draw [thick] (2,0) -- (0.7,-.45);
    \node at (-.5,-1.5) {\scriptsize{$h$}};
    \node at (0,-3.5) {\footnotesize\(p\)};
    \node at (0,0.5) {\footnotesize\(i\)};
    \node at (-2,0.5) {\footnotesize\(i\)};
    \node at (2,0.5) {\footnotesize\(i\)};
     \fill[gray!50] (-1,-1) circle (0.4cm);
    \draw [thick] (-1,-1) circle (0.4cm);
    \node at (-0.98,-1) {\footnotesize$\mu_2$};
    \fill[gray!50] (0,-2) circle (0.4cm);
    \draw [thick] (0,-2) circle (0.4cm);
    \node at (0.02,-2) {\footnotesize$\mu_2$};
  \end{tikzpicture}\\
    &\qquad+\begin{tikzpicture}[scale=0.6,baseline={([yshift=-1ex]current bounding box.center)}]
    \draw [thick] (-2,0) -- (-0.6,-1.4);
    \draw [thick] (-0.35,-1.65) -- (0,-2);
    \draw [thick] (2,0) -- (0,-2);
    \draw [thick] (0,-2) -- (0,-3);
    \draw [thick] (-1,-1)--(0,0);
    \node at (-.5,-1.5) {\scriptsize{$b$}};
    \node at (0,-3.5) {\footnotesize\(p\)};
    \node at (0,0.5) {\footnotesize\(i\)};
    \node at (-2,0.5) {\footnotesize\(i\)};
    \node at (2,0.5) {\footnotesize\(i\)};    
     \fill[gray!50] (-1,-1) circle (0.4cm);
    \draw [thick] (-1,-1) circle (0.4cm);
    \node at (-0.98,-1) {\footnotesize$\mu_2$};
    \fill[gray!50] (0,-2) circle (0.4cm);
    \draw [thick] (0,-2) circle (0.4cm);
    \node at (0.02,-2) {\footnotesize$\mu_2$};
  \end{tikzpicture}
  -
  \begin{tikzpicture}[scale=0.6,baseline={([yshift=-1ex]current bounding box.center)}]
    \draw [thick] (-2,0) -- (0,-2);
    \draw [thick] (2,0) -- (0.65,-1.35);
    \draw [thick] (0,-2) -- (0.4,-1.6);
    \draw [thick] (0,-2) -- (0,-3);
    \draw [thick] (1,-1)--(0,0);
    \node at (0.55,-1.45) {\scriptsize{$b$}};
     \fill[gray!50] (1,-1) circle (0.4cm);
    \draw [thick] (1,-1) circle (0.4cm);
    \node at (1.02,-1) {\footnotesize$\mu_2$};
    \fill[gray!50] (0,-2) circle (0.4cm);
    \draw [thick] (0,-2) circle (0.4cm);
    \node at (0.02,-2) {\footnotesize$\mu_2$};\node at (0,-3.5) {\footnotesize\(p\)};
    \node at (0,0.5) {\footnotesize\(i\)};
    \node at (-2,0.5) {\footnotesize\(i\)};
    \node at (2,0.5) {\footnotesize\(i\)};
  \end{tikzpicture}
  -
  \begin{tikzpicture}[scale=0.6,baseline={([yshift=-1ex]current bounding box.center)}]
    \draw [thick] (-2,0) -- (-0.6,-1.4);
    \draw [thick] (-0.35,-1.65) -- (0,-2);
    \draw [thick] (1,-1) -- (0,-2);
    \draw [thick] (0,-3) -- (0,-2);
    \draw [thick] (1,-1)--(0,0);
    \draw [thick] (-1,-1) -- (0.3,-0.57);
    \draw [thick] (2,0) -- (0.7,-.45);
    \node at (-.5,-1.5) {\scriptsize{$b$}};
    \node at (0,-3.5) {\footnotesize\(p\)};
    \node at (0,0.5) {\footnotesize\(i\)};
    \node at (-2,0.5) {\footnotesize\(i\)};
    \node at (2,0.5) {\footnotesize\(i\)};
     \fill[gray!50] (-1,-1) circle (0.4cm);
    \draw [thick] (-1,-1) circle (0.4cm);
    \node at (-0.98,-1) {\footnotesize$\mu_2$};
    \fill[gray!50] (0,-2) circle (0.4cm);
    \draw [thick] (0,-2) circle (0.4cm);
    \node at (0.02,-2) {\footnotesize$\mu_2$};
  \end{tikzpicture}.
  \end{split}
\end{equation}
As further explained in section \ref{SKamps}, this corresponds to the Schwinger--Keldysh approach to generalised cut amplitudes developed in \cite{Caron-Huot:2023vxl}.

The first few orders of the quasi-isomorphism \(\rho\colon\tilde{\mathfrak L}\to \fL^\cm\) from the \(L_\infty\)-algebra of generalised cut amplitudes to that of ordinary amplitudes given by \cref{thm:keldysh-schwinger-is-quasi-iso} are
\begin{subequations}\label{eq:explicit-formulas}
\begin{align}
  \rho_1(\varphi)          & =pi(\varphi)=\varphi,                                                                  \\
  \rho_2(\varphi_1,\varphi_2) & =p b \mu_2(i(\varphi_1),i(\varphi_2)).
\end{align}
\end{subequations}

The discussion thus far has been restricted to the tree level analysis. To generalise to the loop level one transitions to \emph{quantum} or \emph{loop} $L_\infty$-algebras \cite{Zwiebach:1992ie,Markl:1997bj,Pulmann:2016aa,Doubek:2017naz, Jurco:2019yfd,Saemann:2020oyz, jurčo2024lagrangian}. This, in turn, corresponds to  including loops (formally, i.e.~prior to regularisation and renormalisation) in the BV formalism.  One simply adds\footnote{When the classical BV action obeys the quantum master equation, which is the case in many physically relevant examples.} the   BV Laplacian, which is given by\footnote{In the dual graded manifold picture.} 
\begin{equation}
  \Delta = 
(-1)^{\left|\Phi_i\right|+1} \frac{\overleftarrow{\delta}}{\delta \Phi_i} \frac{\overleftarrow{\delta}}{\delta \Phi_i^{+}},
\end{equation}
to the BV differential. It is then straightforward to construct loops from trees using  homotopy transfer perturbed by $\hbar\Delta$, as described in \cite{Pulmann:2016aa,Doubek:2017naz, Jurco:2019yfd,Saemann:2020oyz}. Diagrammatically, this amounts to allowing for loops in the usual manner for the physically relevant cases, cf.~\cite{Saemann:2020oyz}. Consequently, the generalised cut amplitudes at loop-level are immediately recovered from the perturbed deformation retract data; one simply includes the cuts on loop diagrams in precisely the expected manner.

\section{Keld-ish amplitudes: a toy example}
For a typical BV action with corresponding cyclic $L_\infty$-algebra $\fL=(V, \mu_k,\langle-,-\rangle)$, the cohomology  $\operatorname H^\bullet(V)$ is only nontrivial in degree \(1\), where the on-shell fields live,  and degree \(2\), where the corresponding on-shell antifields reside.
Thus,
$\operatorname H^\bullet(V)=\operatorname H^1(V)\oplus \operatorname H^2(V)$.
The cyclic structure of degree \(-3\) pairs \(V^1\) and \(V^2\) and, thus, provides an isomorphism \(V^1\cong(V^2)^*\) (ignoring the usual functional-analytic subtleties associated with infinitely many dimensions). However, in most physical examples, \(V^1\) comes with an inner product, so that we may identify \(V^1\cong V^2\).
This descends to the cohomology, so that we have the isomorphism \(s\colon\operatorname H^2( V)\to\operatorname H^1( V)\), which we may regard as the degree-shift map upon identifying \(V^1\cong V^2\).

In this generic setting, the simplest possible example of generalised cut amplitudes is given by the  generalised cut propagator  \(\beta=s\), so that $b= i sp$ first projects antifields  to their on-shell component, then degree-shifts them to on-shell fields and then includes them in the full space of fields. 

The resulting generalised cut amplitudes are reminiscent of the time-folded Schwinger--Keldysh amplitudes; let us call them \emph{Keld-ish amplitudes} for brevity. These correspond to Feynman diagrams in which every internal edge can be cut, but---unlike genuine Schwinger--Keldish amplitudes---there are no labels for vertices or external legs representing the time-folds.

\subsection{The Keld-ish amplitudes of scalar field theory}

To give a familiar example of the Keld-ish amplitudes, let us consider scalar field theory with classical action
\begin{equation}\label{scalaraction}
    S[\varphi]\coloneqq\int_{\mathbb{R}^{n}} \mathrm{~d}^n x\left(\frac{1}{2} \varphi\left(-\square-m^2\right) \varphi-\frac{\lambda_3}{3!} \varphi^3-\frac{\lambda_4}{4!} \varphi^4+\dotsb\right).
\end{equation}
Since the gauge symmetry is trivial, the BV fields consist only of the physical scalar field $\varphi\in C^\infty(\mathbb{R}^n)$ and its antifield $\varphi^+\in C^\infty(\mathbb{R}^n)$ and $S[\varphi]=S_{\text{BV}}[\varphi]$.

Let us now briefly recall the homotopy-algebraic perspective. See \cite{Macrelli:2019afx} for full details. The corresponding $L_\infty$-algebra $\fL_\varphi$  has graded vector space $V_\varphi=C^\infty(\mathbb{R}^n)[-1]\oplus C^\infty(\mathbb{R}^n)[-2]$ of fields and antifields, respectively. The only nontrivial products map fields to antifields, 
\begin{equation}
\mu_k \colon (C^\infty(\mathbb{R}^n)[-1])^{\wedge k}\to C^\infty(\mathbb{R}^n)[-2],
\end{equation}
and may be read-off directly from \eqref{scalaraction} as

\begin{equation}\label{scalarproducts}
\mu_1\left(\varphi_1\right):=\left(-\square-m^2\right) \varphi_1, \quad \mu_k\left(\varphi_1, \ldots, \varphi_k\right):=-\lambda_{k-1} \varphi_1\cdots  \varphi_k,
\end{equation}
where the right-hand sides are understood to belong to the space of antifields $C^\infty(\mathbb{R}^n)[-2]$.

The underlying cochain complex $(V_\varphi, \mu_1)$ is
\begin{equation}\label{eq:scalarmodel}
  \begin{tikzcd}[ampersand replacement=\&]
	0 \&\& {C^\infty(\mathbb{R}^n)[-1]} \&\& {} \& {C^\infty(\mathbb{R}^n)[-2]} \&\& 0
	\arrow["0", from=1-1, to=1-3]
	\arrow["{\mu_1\,:=\,-\Box-m^2}", from=1-3, to=1-6]
	\arrow["0", from=1-6, to=1-8]
  \end{tikzcd}
\end{equation}
and the  cohomology $\operatorname H^\bullet(V_\varphi)$ consists of on-shell fields and antifields.  The canonical homotopy $h$ is only nontrivial on the  antifields,  $h: C^\infty(\mathbb{R}^n)[-2]\to C^\infty(\mathbb{R}^n)[-1]$. It is given by  the Feynman propagator\footnote{We leave the degree-shift isomorphism implicit.}
\begin{equation}
    h=\frac1{p^2+m^2-\mathrm i\epsilon}. 
\end{equation}
 The expected $S$-matrix then follows from the canonical minimal model  $\fL_\varphi^\cm$ once the cyclic structure is defined appropriately \cite{Macrelli:2019afx}.

It follows immediately from \cref{recursion} that the Keld-ish amplitudes are computed by the usual Feynman rules except that the propagator \(h\) is modified to be
\begin{equation}
  \tilde h = \frac1{p^2+m^2-\mathrm i\epsilon}+\delta(p^2+m^2),
\end{equation}
where the second term is the generalised cut propagator $b$.
Notice that, unlike ordinary Keldysh--Schwinger amplitudes, there are no labels on the  vertices (which are merely given by \eqref{scalarproducts}), and the delta function picks out both branches of the mass shell, not just the positive-energy shell. (Since there are no labels, there is no canonical way to establish orientation of the momenta so as to fix which is of positive energy and which is of negative energy.)

Thus, the Keld-ish amplitudes will be constructed from Feynman diagrams of the form depicted in \cref{eq:mu3}.
For example, at four points, we have

\begin{equation}
  \begin{split}
    \label{eq:fourpoints}
        \begin{tikzpicture}[scale=0.4,baseline={([yshift=-0.5ex]current bounding box.center)}]
      \draw [thick] (-1.5,0) -- (1.5,-3);
      \draw [thick] (-1.5,-3) -- (1.5,0);
    \end{tikzpicture}
    +
    \begin{tikzpicture}[scale=0.4,baseline={([yshift=-0.5ex]current bounding box.center)}]
      \draw [thick] (-1,0) -- (0,-1);
      \draw [thick] (0,-1) -- (0,-2);
      \draw [thick] (1,0) -- (0,-1);
      \draw [thick] (-1,-3) -- (0,-2);
      \draw [thick] (1,-3) -- (0,-2);
    \end{tikzpicture}
    +
    \begin{tikzpicture}[scale=0.4,baseline={([yshift=-0.5ex]current bounding box.center)}]
      \draw [thick] (-1,0) -- (0,-1);
      \draw [thick,middlearrow={|}] (0,-1) -- (0,-2);
      \draw [thick] (1,0) -- (0,-1);
      \draw [thick] (-1,-3) -- (0,-2);
      \draw [thick] (1,-3) -- (0,-2);
    \end{tikzpicture}
    +
    \begin{tikzpicture}[scale=0.4,baseline={([yshift=-0.5ex]current bounding box.center)}]
      \draw [thick] (-1,1) -- (0,0);
      \draw [thick] (-1,-1) -- (0,0);
      \draw [thick] (0,0) -- (1,0);
      \draw [thick] (1,0) -- (2,1);
      \draw [thick] (1,0) -- (2,-1);
    \end{tikzpicture}
    +
    \begin{tikzpicture}[scale=0.4,baseline={([yshift=-0.5ex]current bounding box.center)}]
      \draw [thick] (-1,1) -- (0,0);
      \draw [thick] (-1,-1) -- (0,0);
      \draw [thick, middlearrow={|}] (0,0) -- (1,0);
      \draw [thick] (1,0) -- (2,1);
      \draw [thick] (1,0) -- (2,-1);
    \end{tikzpicture}
    +
    \begin{tikzpicture}[scale=0.4,baseline={([yshift=-0.5ex]current bounding box.center)}]
      \draw [thick] (2,1) -- (1,0.5);
      \draw [thick] (2,-1) -- (1,-0.5);
      \draw [thick] (1,0.5) -- (1,-0.5);
      \draw [thick] (1,0.5) -- (0.43,0.0725);
      \draw [thick] (-1,-1) -- (0.23,-0.0775);
      \draw [thick] (1,-0.5) -- (-1,1);
    \end{tikzpicture}
    +
    \begin{tikzpicture}[scale=0.4,baseline={([yshift=-0.5ex]current bounding box.center)}]
      \draw [thick] (2,1) -- (1,0.5);
      \draw [thick] (2,-1) -- (1,-0.5);
      \draw [thick, middlearrow={|}] (1,0.5) -- (1,-0.5);
      \draw [thick] (1,0.5) -- (0.43,0.0725);
      \draw [thick] (-1,-1) -- (0.23,-0.0775);
      \draw [thick] (1,-0.5) -- (-1,1);
    \end{tikzpicture}
    +
    \text{loops},
  \end{split}
\end{equation}
where the cuts on internal legs means apply $b$. Hence, the tree-level four-point Keld-ish amplitude is
\begin{equation}
A_{\text{Keld-ish}, 4}^{\text{tree-level}} = \lambda_3^2
  \left(\frac1{s+m^2}+\frac1{t+m^2}+\frac1{u+m^2}+
  \delta(s+m^2)+
  \delta(t+m^2)+
  \delta(u+m^2)
  \right)
  +\lambda_4,
\end{equation}
where \(s\), \(t\), and \(u\) are the Mandelstam variables.
Thus we retrieve the usual expression for the cut amplitudes in the sense of  \cite[\S V.A]{britto2024cuttingedge}.

\section{Schwinger--Keldysh amplitudes as generalised cut amplitudes}\label{SKamps}

As alluded to throughout the paper, the Schwinger--Keldysh amplitudes fit into the  homotopy-algebraic formalism of generalised cuts. We will in this section make this statement precise and work out explicitly the expression for the corresponding generalised cut propagator $b_{\text{SK}}$.
\subsection{Review of the Schwinger--Keldysh formalism}\label{ssec:path_integral}
Recall that the scattering amplitudes of quantum field theory are given by the Lehmann--Symanzik--Zimmermann reduction of time-ordered correlators of local operators.
In the Keldysh--Schwinger formalism, we instead look at correlators of local operators ordered in different ways and the Lehmann--Symanzik--Zimmermann reductions thereof.
Specifically, let us consider correlators of the ansatz
\begin{equation}\label{eq:time-ordering-ansatz}
  \left\langle \dotsm \mathcal T\mleft\{\varphi_{i_3} \dotsm \varphi_{i_2+1}\mright\} \overline{\mathcal T}\mleft\{\varphi_{i_2} \dotsm \varphi_{i_1+1}\mright\} \mathcal T\mleft\{\varphi_{i_1} \dotsm \varphi_1\mright\}\right\rangle,
\end{equation}
where \(\mathcal T\) denotes time ordering and \(\overline{\mathcal T}\) denotes anti-time ordering (i.e.~in the opposite order as the usual time ordering).\footnote{Such an ansatz is universal in that, using identities such as
\begin{equation}
  \begin{aligned}
    \overline{\mathcal T}\mleft\{\varphi_1 \varphi_2\mright\} & \coloneqq \varphi_1 \varphi_2 \Theta\mleft(x_2^0-x_1^0\mright)+\varphi_2 \varphi_1 \Theta\mleft(x_1^0-x_2^0\mright)                                                                             \\
                                                        & =-\mathcal T\mleft\{\varphi_1 \varphi_2\mright\}+\mathcal T\mleft\{\varphi_1\mright\} \mathcal T\mleft\{\varphi_2\mright\}+\mathcal T\mleft\{\varphi_2\mright\} \mathcal T\mleft\{\varphi_1\mright\},
  \end{aligned}
\end{equation}
any product of time-ordered, anti-time-ordered, or unordered operators can be expressed as a linear combination of correlators of the form \eqref{eq:time-ordering-ansatz}.}
One may think of such a correlator as path-ordered along a contour \(\mathcal C\) that starts off from the far past, goes to the far future (the operators \(\varphi_{i_1}\dotsm\varphi_1\) are inserted here), and then turns around and goes back in time to the far past (the operators \(\varphi_{i_2}\dotsm\varphi_{i_1+1}\) are inserted here), then turns around and goes back to the future, and so on. Thus, \(\mathcal C\) consists of multiple segments that we shall call `time folds', either oriented forward (i.e.\ from the far past to the far future) or oriented backward (i.e.\ from the far future to the far past).\footnote{Technically, one should define \(\mathcal C\) in complexified time coordinates. For more details, see e.g.~\cite{cond-mat/0506130}.}

Such correlators can be computed by path integrals along the contour $\mathcal{C}$:
\begin{equation}\label{timecontourintegral}
  Z_{\text{SK}}[J] \coloneqq\int\mathrm D\phi\,{\mathcal{C}} \exp \left(\mathrm i \int_{\mathcal{C}}\mathrm dt\int\mathrm d^{d-1}\vec x\, \mathcal{L}[\varphi(t,\vec x)]+{J}(t,\vec x) \varphi(t,\vec x)\right),
\end{equation}
where \(\mathcal C\exp\) is the path-ordered exponential along \(\mathcal C\). 
If the contour \(\mathcal C\) consists of $N$ time folds, writing the source \(J\) as one field \(J_{(i)}\) for each time segment \((i)\), this may be written as 
\begin{equation}
\begin{aligned}
  \MoveEqLeft Z_{S K}\left[{J}_{(1)}, \ldots {J}_{(N)}\right]\\&=
  \int\mathrm D\phi\,{\mathcal{C}} \exp\left(\mathrm i\sum_{i=1}^{N}(-1)^{i+1} \int\mathrm d^dx\, \left(\mathcal{L}\left[\varphi(x)\right]-{J}_{(i)}(x)\varphi(x)\right)\right)\\
  &=
  \int\mathrm D\phi_{(1)}\dotsm\mathrm D\phi_{(N)}\,{\mathcal T} \exp\left(\mathrm i\int\mathrm d^dx\sum_{i=1}^{N}(-1)^{i+1} \left(\mathcal{L}\left[\varphi_{(i)}(x)\right]-{J}_{(i)}(x)\varphi_{(i)}(x)\right)\right).
\end{aligned}
\end{equation}
Now, the path integral looks like a normal time-ordered path integral --- but with \(N\) copies of the fields, of which  of which \(\lfloor N/2\rfloor\) come with wrong-sign actions. That is, of the \(N\) copies of the fields, \(\lceil N/2\rceil\) travel forward in time and \(\lfloor N/2\rfloor\), depending on which time fold they live on.

The \emph{Schwinger--Keldysh amplitudes} are then obtained by taking a Lehmann--Symanzik--Zimmermann reduction of such non-time-ordered correlators. Such `amplitudes' can be obtained by a set of Feynman rules given in \cite{VELTMAN1963186,2008.11730,Caron-Huot:2023vxl} as follows. Suppose that we have \(N\) time folds labelled $(1), (2),\dotsc, (N)$, where the time folds with odd labels $(2n+1)$ are forward, those with even labels $(2n)$ are backward, and the path is $(1)$--$(2)$--$(3)\dotsb (N)$.  To avoid confusion between time folds and other common labels, following \cite{Caron-Huot:2023vxl} we will also use  I, II, III, IV, \ldots\ to denote $(1), (2), (3), (4)\dotsc$.  Then the rules are:
\begin{itemize}
  \item All vertices are labelled I, II, III, IV, etc. 
  \item The even-label vertices come with extra minus sign.
  \item For propagators between vertices of the same odd label, the propagator is
        \begin{equation}\frac{1}{p^2+m^2-\mathrm i\epsilon}\end{equation}
        as usual.
  \item For propagators between vertices of the same even label, the propagator is
        \begin{equation}-\frac1{p^2+m^2+\mathrm i\epsilon}.\end{equation}
        Notice that, in addition to the overall sign, there is an additional sign in front of the \(\epsilon\).
  \item Propagator between different adjacent labels is \(2\pi\mathrm i\delta(p^2+m^2)\Theta(p^0)\) where \(p\) flows in increasing-label direction.
\end{itemize}

\subsection{Schwinger--Keldysh amplitudes as generalised cut amplitudes}
Given the form of the Schwinger--Keldysh Feynman rules, it is plausible that Schwinger--Keldysh amplitudes are an instance of generalised cut amplitudes in our sense and hence are quasi-isomorphic to ordinary amplitudes. We now make this explicit.

\paragraph{\(\boldsymbol N\)-tupled theory}

To be concrete, suppose that the original theory is classically described by a cyclic \(L_\infty\)-algebra \(( V,\mu_k,\langle-,-\rangle)\), and we are considering Schwinger--Keldysh amplitudes with labels \(1,\dotsc,N\). As remarked in \cref{ssec:path_integral}, at any point in time, we have \(N\) copies of the field content of the theory, of which \(\lceil N/2\rceil\) travel forward in time and \(\lfloor N/2\rfloor\) backward. In the homotopy-algebraic formalism this translates into the cyclic  $L_\infty$-algebra 
\begin{equation}
  \fL_{N}=(V_{N}, \mu_k^{N}, \langle-,-\rangle_{N}).
\end{equation}
 The underlying graded vector space \( V_{N}\) is concretely
\begin{equation}
  V_{N}= V\otimes\mathbb R^N= V^{(1)}\oplus V^{(2)}\oplus\dotsb\oplus V^{(N)},
\end{equation}
where \( V^{(1)},\dotsc, V^{(N)}\) are \(N\) isomorphic copies of the underlying graded vector space of \( V\). The cyclic structure and $L_\infty$-algebra operations are given, respectively, by 
\begin{equation}
  \langle-,-\rangle_N=\sum_{i=1}^N \langle-,-\rangle^{(i)}, \quad
  \mu_{k}^{N} = \sum_{i=1}^N(-1)^{i-1}\mu_k^{(i)},
\end{equation}
where $\langle-,-\rangle^{(i)}$, and $\mu_k^{(i)}$ are the cyclic structure and $\mu_k$ of the original theory,\footnote{We leave implicit the embedding of $V^{(i)}$ into $V_N$.} but restricted to $V^{(i)}$. \footnote{Technically, we should make sure to obtain the correct sign of \(\epsilon\) in the kinetic operator \(\mu_1=\square-m^2+\mathrm i\epsilon\). Assuming a single massive scalar field for simplicity, we have
  \begin{equation}\label{eq:propagator-prescription}
    \mu_1 = \sum_{i=1}^N(-1)^{i-1}(\square^{(i)}-m^2+(-1)^{i-1}\mathrm i\epsilon)=\sum_{i=1}^N(-1)^{i-1}(\square^{(i)}-m^2)+\mathrm i\epsilon.
  \end{equation}
  That is, we impose a \(+\mathrm i\epsilon\) prescription uniformly.
}
The sign \((-1)^{i-1}\) accounts for the fact that odd values of \(i\) correspond to fields travelling forward in time while even values correspond to fields travelling backward in time.\footnote{The $L_\infty$-algebra $\fL_N$ is thus obtained as the direct sum of $L_\infty$-algebras $\bigoplus_{i=1}^N \fL^{(i)}$ where $\fL^{(i)}\coloneqq(V^{(i)},(-1)^{i-1}\mu_k^{(i)},\langle-,-\rangle^{(i)})$.}

The homotopy Maurer-Cartan action then takes the form
\begin{equation}
  S_{\fL_{N}}=\sum_{i=1}^N(-1)^{i-1}S_{\fL}[\varphi_i],
\end{equation}
where \(\varphi_\text{1},\dotsc,\varphi_N\) are the fields of \( V^{(1)},\dotsc, V^{(N)}\), respectively. The even \(i\) correspond to particles going backwards in time. To compute scattering amplitudes for \( V\), we can use \(N\) copies of the standard propagator:\footnote{to be consistent with \eqref{eq:propagator-prescription}, the propagator \(h\) should always use the \(+\mathrm i\epsilon\) contour.}
\begin{equation}\label{Ntupleh}
  h = \sum_{i=1}^N(-1)^{i-1}h^{(i)}.
\end{equation}

The tree and loop scattering amplitudes for \( V_N\) are then simply \(N\) copies of those of \( V\); there is no new information. The allowed Feynman diagrams include
\begin{equation}
  \begin{tikzpicture}[scale=0.7,baseline={([yshift=-0.5ex]current bounding box.center)}]
    \draw [thick] (-1,1) -- (0,0);
    \draw [thick] (-1,-1) -- (0,0);
    \draw [thick] (1,1) -- (0,0);
    \draw [thick] (1,-1) -- (0,0);
    \filldraw [gray] (0,0) circle (15pt);
    \draw [black, thick] (0,0) circle (15pt);
    \node at (-1.5,1.43) {I};
    \node at (1.5,1.43) {I};
    \node at (-1.5,-1.43) {I};
    \node at (1.5,-1.43) {I};
  \end{tikzpicture}
  \qquad\text{and}\qquad
  \begin{tikzpicture}[scale=0.7,baseline={([yshift=-0.5ex]current bounding box.center)}]
    \draw [thick] (-1,1) -- (0,0);
    \draw [thick] (-1,-1) -- (0,0);
    \draw [thick] (1,1) -- (0,0);
    \draw [thick] (1,-1) -- (0,0);
    \filldraw [gray] (0,0) circle (15pt);
    \draw [black,thick ] (0,0) circle (15pt);
    \node at (-1.5,1.43) {II};
    \node at (1.5,1.43) {II};
    \node at (-1.5,-1.43) {II};
    \node at (1.5,-1.43) {II};
  \end{tikzpicture},
\end{equation}
but there are no connected diagrams connecting external legs with different labels, since there is no propagator that connects vertices of different labels.

\paragraph{Generalised cut propagator} The $N$-tupled theory merely generates labelled identical copies of the original $S$-matrix, with no diagrams connecting them. This corresponds to the canonical $S$-matrix of the $N$-tupled theory produced by homotopy transfer using the canonical propagator \eqref{Ntupleh}.

To obtain the Schwinger--Keldysh diagrams, we introduce a generalised cut propagator to implement the cut edges. In this case, we have
\begin{equation}
  b= 2\pi\mathrm i\sum_{i=1}^{N-1}\left(\delta(p^2+m^2)\Theta(p^0)\sigma_{i\rightarrow i+1}+\delta(p^2+m^2)\Theta(-p^0)\sigma_{i+1\rightarrow i}\right)
\end{equation}
where \(\sigma_{i\to j}\) is the degree $-1$ operator that shifts the label from \(i\) to \(j\). Notice that we have two cases, depending on whether the momentum is flowing \(i\to i+1\) (in which case we require \(p^0>0\)) or \(i+1\to i\) (in which case we require \(p^0<0\)).

Then it is manifest from inspection that the above \(b\) implements the Feynman rules that define the Schwinger--Keldysh amplitudes. For example
\begin{equation}
  \begin{tikzpicture}[scale=0.7,baseline={([yshift=-0.5ex]current bounding box.center)}]
    \draw [thick]  (-1,1) -- (0,0);
    \draw [thick]  (-1,-1) -- (0,0);
    \draw [thick]  (1,1) -- (0,0);
    \draw [thick]  (1,-1) -- (0,0);
    \filldraw [gray] (0,0) circle (15pt);
    \draw [thick] (0,0) circle (15pt);
    \node at (-1.5,1.43) {II};
    \node at (1.5,1.43) {II};
    \node at (-1.4,-1.43) {I};
    \node at (1.4,-1.43) {I};
  \end{tikzpicture}
  =
  \begin{tikzpicture}[scale=0.5,baseline={([yshift=-0.5ex]current bounding box.center)}]
    \draw [thick]  (-1,0) -- (0,-1);
    \draw [thick,middlearrow={|}] (0,-1) -- (0,-2);
    \draw [thick]  (1,0) -- (0,-1);
    \draw [thick] (-1,-3) -- (0,-2);
    \draw [thick]  (1,-3) -- (0,-2);
    \node at (-1.5,0.5) {II};
    \node at (1.5,0.5) {II};
    \node at (-1.5,-3.5) {I};
    \node at (1.4,-3.5) {I};
  \end{tikzpicture}
  +
  \begin{tikzpicture}[scale=0.5,baseline={([yshift=-0.5ex]current bounding box.center)}]
    \draw [thick, middlearrow={|}]  (0,-1) -- (0,-2);
    \draw [thick, middlearrow={|}]  (1,-1) -- (1,-2);
    \draw [thick]  (0,-1) -- (1,-1);
    \draw [thick]  (0,-2) -- (1,-2);
    \draw [thick] (-1,0) -- (0,-1);
    \draw [thick] (1,-1) -- (2,0);
    \draw [thick] (-1,-3) -- (0,-2);
    \draw [thick] (2,-3) -- (1,-2);
    \node at (-1.5,0.5) {II};
    \node at (2.5,0.5) {II};
    \node at (-1.5,-3.5) {I};
    \node at (2.4,-3.5) {I};
  \end{tikzpicture}
  +\; \dotsb.
\end{equation}
From the general results of section \eqref{recursion}, it immediately follows that the Schwinger--Keldysh amplitudes are entirely equivalent to (multiple disjoint copies of) the $S$-matrix. That is, there exists a quasi-isomorphism that expresses every Schwinger--Keldysh amplitude in terms of the $S$-matrix and vice versa. Applying  \eqref{eq:trivialretract} we obtain the perturbative  blobology of \cite{Caron-Huot:2023vxl}, while using \eqref{eq:tildefromog} yields the equivalent Schwinger--Keldysh approach.

\subsection{The one-particle stability axiom}\label{sec:stabilityaxiom}
In this section (and this section only), all propagators are oriented in such a way that time flows to the left and that a propagator always has positive energy flowing to the left.

The axiomatics of \cite{Caron-Huot:2023vxl} contain an ingredient not found in earlier discussions \cite{2008.11730,VELTMAN1963186}, which is the explicit requirement that the one-particle state be stable. Using this requirement, \cite[(4.55)]{Caron-Huot:2023vxl} asserts that
\begin{subequations}
  \begin{align}
    \begin{tikzpicture}[scale=0.7,baseline={([yshift=-0.5ex]current bounding box.center)}]
      \draw [thick,middlearrow={latex reversed}] (0,0) -- (1,1);
      \draw [thick,middlearrow={latex reversed}] (0,2)--  (1,1);
      \draw [thick,middlearrow={|}] (1,1) -- (2,1);
      \draw [thick,middlearrow={latex reversed}] (2,1) -- (3,0);
      \draw [thick,middlearrow={latex reversed}] (2.55,1.55) -- (3,2);
      \draw [thick,middlearrow={latex reversed}] (2,1) -- (2.55,1.55);
      \draw [thick,middlearrow={latex reversed}]  (2.1,2) -- (2.55,1.55);
      \node at (0,0) [anchor=north east] {II};
      \node at (0,2) [anchor=south east] {II};
      \node at (3,0) [anchor=north west] {I};
      \node at (2,2) [anchor=south east] {I};
      \node at (3,2) [anchor=south west] {I};
    \end{tikzpicture}
      & \ne0 \label{eq:doesnt_vanish} \\
    \begin{tikzpicture}[scale=0.7,baseline={([yshift=-0.5ex]current bounding box.center)}]
      \draw [thick,middlearrow={latex reversed}] (0,0) -- (1,1);
      \draw [thick,middlearrow={latex reversed}] (0,2)--  (1,1);
      \draw [thick,middlearrow={latex reversed}] (1,1) -- (2,1);
      \draw [thick,middlearrow={latex reversed}] (2,1) -- (3,0);
      \draw [thick,middlearrow={latex reversed}] (2.55,1.55) -- (3,2);
      \draw [thick,middlearrow={|}] (2,1) -- (2.55,1.55);
      \draw [thick,middlearrow={latex reversed}]  (2.1,2) -- (2.55,1.55);
      \node at (0,0) [anchor=north east] {II};
      \node at (0,2) [anchor=south east] {II};
      \node at (3,0) [anchor=north west] {II};
      \node at (2,2) [anchor=south east] {I};
      \node at (3,2) [anchor=south west] {I};
    \end{tikzpicture}
      & =0 \label{eq:does_vanish}
  \end{align}
\end{subequations}
even though the latter does not follow from the mere Feynman rules for Schwinger--Keldysh amplitudes.

\paragraph{Effective actions and $L_\infty$-algebras}
To explain this diagrammatically, we first recall what the \(L_\infty\)-algebra scattering amplitude formalism purports to compute: given an action functional \(S[\varphi]\) for a field \(\varphi\) (or, equivalently, the cyclic \(L_\infty\)-algebra \( \fL\) whose homotopy Maurer--Cartan action is \(S\)), it computes the \(S\)-matrix order-by-order \emph{provided} that the asymptotic in- and out-states are the solutions to the linearised classical equations of motion (i.e.\ the cohomology \(\operatorname H^\bullet( V)\)).

This assumption is crucial. It is a commonplace that the fields in terms of which an action is written are often \emph{not} the elementary fields of the asymptotic Fock spaces in the far past or far future. For instance, the spectrum of quantum chromodynamics (at zero temperature) does \emph{not} consist of quarks and gluons but rather of stable hadrons. Less dramatically, it may happen that one of the elementary fields that appear in the action may not be stable and decay. In these cases, one can in principle always write down an infrared action in terms of the stable elementary particles; by \emph{effective actions} we mean such an action (and not necessarily in the Wilsonian or one-particle-irreducible senses).

A necessary condition for an \(L_\infty\)-algebra to encode an effective action is that there be no nontrivial amplitudes with only one incoming or only one outgoing leg; in particular, all three-point amplitudes must vanish. For example, a phenomenological action might have an elementary field for the neutral pion and a vertex \(AA\pi^0\) that encodes the decay mode \(\pi^0\to\gamma\gamma\) of the neutral pion into a pair of photons, corresponding to a nontrivial three-point amplitude; this action is \emph{not} an effective action since it contains \(\pi^0\) (whose excitations are not part of the asymptotic Fock space) as an elementary field.

This necessary condition can be formalised as saying that, for the minimal model,
\begin{multline}
  \langle \varphi_0,\mu_n^{\cm}(\varphi_1,\dotsc,\varphi_n)\rangle^\cm=0\quad\text{if }{\Big\rvert}\left\{i\in\{0,\dotsc,n\}\middle| p_0^{(i)}>0\right\}{\Big\rvert} \le1\\\text{or }{\Big\rvert}\left\{i\in\{0,\dotsc,n\}\middle| p_0^{(i)}<0\right\}{\Big\rvert} \le1
\end{multline}
where \(p^{(i)}\) is the momentum of \(\varphi_i\) (all regarded as incoming by using cyclicity).\footnote{We ignore the degenerate case where the particle energy is zero for simplicity.} As a consequence,
\begin{equation}
  \mu_2^{\cm}=0.
\end{equation}

To reproduce the \(S\)-matrix-theoretic discussion in \cite{Caron-Huot:2023vxl}, we should only apply the \(L_\infty\)-algebraic formalism to effective actions. This then explains \eqref{eq:does_vanish}
since the upper right corner factorises into an amplitude in which one incoming particle decays into a number of particles which violates the above necessary condition.

On the other hand, what are we to make of the other diagram \eqref{eq:doesnt_vanish}? Superficially, the left half of the diagram appears as a decay process of a single particle; should it not also vanish?

The resolution is that the internal legs in the diagrams appearing in \cite{Caron-Huot:2023vxl} sometimes in fact do represent unstable particles, which are then treated as if they were stable by using the narrow-width approximation. On the other hand, when using the effective action, there are simply no lines corresponding to unstable particles; all `unstable particles' in fact represent multiparticle states that appear in cut loop diagrams. So \eqref{eq:doesnt_vanish} actually is short for
\begin{equation}
  \begin{tikzpicture}[scale=0.7,baseline={([yshift=-0.5ex]current bounding box.center)}]
    \draw [thick,middlearrow={latex reversed}] (0,0) -- (1,1);
    \draw [thick,middlearrow={latex reversed}] (0,2)--  (1,1);
    \draw [thick,middlearrow={|}] (1,1) -- (2,1);
    \draw [thick,middlearrow={latex reversed}] (2,1) -- (3,0);
    \draw [thick,middlearrow={latex reversed}] (2.55,1.55) -- (3,2);
    \draw [thick,middlearrow={latex reversed}] (2,1) -- (2.55,1.55);
    \draw [thick,middlearrow={latex reversed}]  (2.1,2) -- (2.55,1.55);
    \node at (0,0) [anchor=north east] {II};
    \node at (0,2) [anchor=south east] {II};
    \node at (3,0) [anchor=north west] {I};
    \node at (2,2) [anchor=south east] {I};
    \node at (3,2) [anchor=south west] {I};
  \end{tikzpicture}
  +
  \begin{tikzpicture}[scale=0.7,baseline={([yshift=-0.5ex]current bounding box.center)}]
    \draw [thick,middlearrow={latex reversed}] (0,0) -- (1,1);
    \draw [thick,middlearrow={latex reversed}] (0,2)--  (1,1);
    \draw [thick, middlearrow={|}] (1,1) to[in=-135, out =-45] (2,1);
    \draw [thick, middlearrow={|}] (1,1) to[in=135, out =45] (2,1);
    \draw [thick,middlearrow={latex reversed}] (2,1) -- (3,0);
    \draw [thick,middlearrow={latex reversed}] (2.55,1.55) -- (3,2);
    \draw [thick,middlearrow={latex reversed}] (2,1) -- (2.55,1.55);
    \draw [thick,middlearrow={latex reversed}]  (2.1,2) -- (2.55,1.55);
    \node at (0,0) [anchor=north east] {II};
    \node at (0,2) [anchor=south east] {II};
    \node at (3,0) [anchor=north west] {I};
    \node at (2,2) [anchor=south east] {I};
    \node at (3,2) [anchor=south west] {I};
  \end{tikzpicture}
  +    \begin{tikzpicture}[scale=0.7,baseline={([yshift=-0.5ex]current bounding box.center)}]
    \draw [thick,middlearrow={latex reversed}] (0,0) -- (1,1);
    \draw [thick,middlearrow={latex reversed}] (0,2)--  (1,1);
    \draw [thick] (1,1) -- (1.3,1);
    \draw [thick, middlearrow={|}] (1.3,1) to[in=90, out =90] (2,1);
    \draw [thick, middlearrow={|}] (1.3,1) to[in=-90, out =-90] (2,1);
    \draw [thick,middlearrow={latex reversed}] (2,1) -- (3,0);
    \draw [thick,middlearrow={latex reversed}] (2.55,1.55) -- (3,2);
    \draw [thick,middlearrow={latex reversed}] (2,1) -- (2.55,1.55);
    \draw [thick,middlearrow={latex reversed}]  (2.1,2) -- (2.55,1.55);
    \node at (0,0) [anchor=north east] {II};
    \node at (0,2) [anchor=south east] {II};
    \node at (3,0) [anchor=north west] {I};
    \node at (2,2) [anchor=south east] {I};
    \node at (3,2) [anchor=south west] {I};
  \end{tikzpicture}
  + \dotsb \neq 0
\end{equation}
where now all lines do represent stable particles, and where indeed the first diagram vanishes but the other ones need not.

\acknowledgments

D.S.H.J. and H.K. were supported by the Leverhulme Research Project Grant RPG-2021-092. The authors thank Christian Saemann, Martin Wolf, and the anonymous referee for helpful comments.

\newcommand\Yu{Yu}
\bibliographystyle{JHEP}
\bibliography{bigone,extra}
\end{document}